\documentclass[9pt]{entcs}

\usepackage{entcsmacro}
\usepackage[english]{babel}

\usepackage{amsmath}
\usepackage{amssymb}
\usepackage{bbold}              %% For lower-case \mathbb
\usepackage{caption}
\usepackage{color}
\usepackage{dashbox}
\usepackage[mathscr]{euscript}  %% For \mathscr

\usepackage{graphicx}
\usepackage{mathrsfs}           %% For \mathscr
\usepackage{multirow}
\usepackage{paralist}
\usepackage{relsize}
\usepackage{stmaryrd}
\usepackage{tikz}
\usetikzlibrary{arrows}
\usetikzlibrary{automata}
\usetikzlibrary{calc}
\usetikzlibrary{cd}
\usetikzlibrary{decorations.markings}
\usetikzlibrary{decorations.pathmorphing}
\usetikzlibrary{positioning}
\usepackage{tikz-qtree}

\usepackage[normalem]{ulem}
\usepackage{xargs}
\usepackage{xifthen}
\usepackage{wrapfig}

\tikzset{
  negated/.style={decoration={markings, mark= at position 0.5 with {\node[transform shape] (tempnode) {$/$};}}, postaction={decorate}},
  Rightarrow/.style={double equal sign distance, double distance=4pt, >={Implies}, ->},
  RightarrowDashed/.style={double equal sign distance, double distance=4pt, >={Implies}, ->, dashed},
  Rrightarrow/.style={-, preaction={draw,Rightarrow}},
  RrightarrowDashed/.style={-, dashed, preaction={draw,RightarrowDashed}},
  equal/.style={double equal sign distance, double distance=4pt},
  equiv/.style={-, preaction={draw,equal}}
}

\definecolor{metacolor}{rgb}{1,0.5,0}
\definecolor{explcolor}{rgb}{0,0.5,1}

\newcommand{\stkout}[1]{\ifmmode\text{\sout{\ensuremath{#1}}}\else\sout{#1}\fi}

\newcommand{\comment}[1]{}

\newcommand{\cf}{\emph{cf.}~}

\newcommand{\ie}{\emph{i.e.}~}
\newcommand{\ih}{\emph{i.h.}~}

\newcommand{\coloneq}{\ensuremath{:=}}
\newcommand{\Coloneq}{\ensuremath{:\coloneq}}

\newcommand{\eqalpha}{\ensuremath{=_{\alpha}}}

\newcommand{\eqdef}{\ensuremath{\triangleq}}

\newcommand{\eqtrans}{\ensuremath{\Mapsto}}

\newcommand{\llbrace}{\{\!\!\{}
\newcommand{\rrbrace}{\}\!\!\}}

\newcommand{\lista}[1]{\ensuremath{[{#1}]}}
\newcommand{\listconc}[2]{\ensuremath{{#1}{#2}}}

\newcommand{\set}[1]{\ensuremath{\{{#1}\}}}

\newcommand{\calcCAP}{\textsf{CAP}}
\newcommand{\calcLambda}{\ensuremath{\lambda}}

\newcommand{\calcLambdaX}{\ensuremath{\lambda}_\mathtt{dB}}
\newcommand{\calcPPC}{\ensuremath{\mathsf{PPC}}}
\newcommand{\calcPPCX}{\ensuremath{\mathsf{PPC}_\mathtt{dB}}}

\newcommand{\rbeta}{\ensuremath{\beta}}
\newcommand{\rbetax}{\ensuremath{\mathtt{dB}}}

\newcommand{\rPPC}{\ensuremath{\mathtt{PPC}}}
\newcommand{\rPPCX}{\ensuremath{\mathtt{dB}}}

\newcommand{\reduce}[1][]{\ensuremath{\rightarrow_{#1}}}

\newcommand{\rrule}[1]{\ensuremath{\mathrel{\mapsto_{#1}}}}

\newcommand{\Matchable}[1]{\ensuremath{{\mathbb{M}_{#1}}}}
\newcommand{\Natural}{\ensuremath{{\mathbb{N}}}}

\newcommandx{\PhiEqup}[1][1=\emptyset]{\ensuremath{\varPhi_{\eqtypeup[#1]}}}

\newcommandx{\PhiSubup}[1][1=\emptyset]{\ensuremath{\varPhi_{\subtypeup[#1]}}}

\newcommand{\Term}[1]{\ensuremath{{\mathbb{T}_{#1}}}}
\newcommand{\TermData}[1]{\ensuremath{{\mathbb{D}_{#1}}}}

\newcommand{\TermVariable}{\ensuremath{{\mathbb{V}}}}

\renewcommand{\S}{\ensuremath{\mathcal{S}}}

\newcommand{\itermabs}{\ensuremath{\lambda}}

\newcommand{\itermapp}{\ensuremath{\,}}

\newcommand{\termabs}[3][]{\ensuremath{\itermabs_{#1}{#2}.{#3}}}
\newcommand{\termabsi}[1]{\ensuremath{\itermabs{#1}}}

\newcommand{\termapp}[2]{\ensuremath{{#1}\itermapp{#2}}}

\newcommandx{\termidx}[2][2=]{\ensuremath{\mathtt{#1}_{\mathtt{#2}}}}
\newcommandx{\termidxm}[2][2=]{\ensuremath{\termidx{\widehat{#1}}[#2]}}
\newcommandx{\termidxv}[2][2=]{\ensuremath{\termidx{#1}[#2]}}
\newcommand{\termmatch}[1]{\ensuremath{\widehat{{#1}}}}

\newcommand{\termpair}[2]{\ensuremath{\langle{#1},{#2}\rangle}}

\newcommand{\termvar}[1]{\ensuremath{{#1}}}

\newcommand{\ctxt}[1]{\ensuremath{\mathtt{#1}}}
\newcommand{\ctxtapply}[2]{\ensuremath{{#1}\langle{#2}\rangle}}

\newcommand{\itypedata}{\ensuremath{\mathbin@}}

\newcommandx{\deriv}[3][1=,3=]{\ensuremath{{#1}\rhd_{#3}{#2}}}

\newcommandx{\envsum}[3][1=,3=]{\ensuremath{{#1}\mathrel{+_{#3}}{#2}}}

\newcommand{\change}[3][]{\ensuremath{{#2}\setminus^{\!{#1}}{#3}}}

\newcommand{\dom}[1]{\ensuremath{\funcapply{\mathsf{dom}}{#1}}}

\newcommand{\fm}[1]{\ensuremath{\funcapply{\mathsf{fm}}{#1}}}

\newcommand{\fsize}[1]{\ensuremath{|{#1}|}}

\newcommand{\funcapply}[2]{\ensuremath{{#1}({#2})}}
\newcommand{\funccomp}[2]{\ensuremath{{#1}\circ{#2}}}
\newcommand{\funcid}{\ensuremath{\mathit{id}}}

\newcommand{\fv}[2][]{\ensuremath{\funcapply{\mathsf{fv}_{#1}}{#2}}}
\newcommandx{\toPPC}[4][1=,2=,4=]{\ensuremath{\llparenthesis{#3}\ifthenelse{\isempty{#4}}{}{,{#4}}\rrparenthesis^{\scriptscriptstyle{#2}}_{\scriptscriptstyle{#1}}}}
\newcommandx{\toPPCX}[4][1=,2=,4=]{\ensuremath{\llbracket{#3}\ifthenelse{\isempty{#4}}{}{,{#4}}\rrbracket^{\scriptscriptstyle{#2}}_{\scriptscriptstyle{#1}}}}
\newcommandx{\toLambda}[1]{\ensuremath{\llparenthesis{#1}\rrparenthesis}}
\newcommandx{\toLambdaX}[1]{\ensuremath{\llbracket{#1}\rrbracket}}

\newcommandx{\fdecv}[3][1=,3=]{\ensuremath{\funcapply{\downarrow^{#3}_{#1}}{#2}}}
\newcommandx{\fdecm}[3][1=,3=]{\ensuremath{\funcapply{\Downarrow^{#3}_{#1}}{#2}}}
\newcommandx{\fincv}[3][1=,3=]{\ensuremath{\funcapply{\uparrow^{#3}_{#1}}{#2}}}
\newcommandx{\fincm}[3][1=,3=]{\ensuremath{\funcapply{\Uparrow^{#3}_{#1}}{#2}}}

\newcommand{\img}[1]{\ensuremath{\funcapply{\mathsf{img}}{#1}}}

\newcommand{\match}[3][]{\ensuremath{{\color{metacolor}\llbrace}\change[#1]{#2}{#3}{\color{metacolor}\rrbrace}}}
\newcommand{\matchfail}{\ensuremath{\mathtt{fail}}}
\newcommand{\matchundet}{\ensuremath{\mathtt{wait}}}
\newcommand{\matchapply}[2]{\ensuremath{{#1}\,{#2}}}
\newcommand{\matchcomp}[2]{\ensuremath{\funccomp{#1}{#2}}}
\newcommand{\matchunion}[2]{\ensuremath{{#1}\uplus{#2}}}

\newcommand{\subs}[2]{\ensuremath{{\color{metacolor}\{}\change{#1}{#2}{\color{metacolor}\}}}}
\newcommand{\subsapply}[2]{\ensuremath{{#1}{#2}}}
\newcommand{\subscomp}[2]{\ensuremath{\funccomp{#1}{#2}}}
\newcommand{\subsid}{\ensuremath{{\color{metacolor}\{\}}}}

\newcommand{\sym}[1]{\ensuremath{\funcapply{\mathsf{sym}}{#1}}}

\newlength{\PNdist}
\newlength{\PNmarg}
\newlength{\PNsize}
\tikzstyle{PNcir}=[circle,draw=black,inner sep=0pt,text width=\PNsize]
\tikzstyle{PNsub}=[rectangle,draw=black,inner sep=5pt,minimum height=\PNsize,rounded corners=8pt]

\tikzstyle{PNardot}=[dashed,rounded corners=8pt]
\tikzstyle{PNarrow}=[rounded corners=8pt]
\tikzstyle{PNarout}=[-{triangle 45},rounded corners=8pt]
\tikzstyle{PNarold}=[-{open triangle 45},rounded corners=8pt]

\newcommand{\subterm}{\ensuremath{\subseteq}}

\newcommand{\subtype}{\ensuremath{\preceq}}
\newcommand{\eqtype}{\ensuremath{\simeq}}

\newcommandx{\subtypeup}[1][1=\emptyset]{\ensuremath{\subtype^{#1}_{\mathfrak{n}}}}
\newcommandx{\eqtypeup}[1][1=\emptyset]{\ensuremath{\eqtype^{#1}_{\mathfrak{n}}}}

\newcommand{\pavoids}[2]{\ensuremath{{#1}\mathrel{\mathtt{avoids}}{#2}}}

\def\lastname{Mart\'{i}n, R\'{i}os, Viso}
\begin{document}
\begin{frontmatter}
  \title{Pure Pattern Calculus \emph{\`{a} la} de Bruijn\thanksref{INFINIS}}
  \author{Alexis Mart\'{i}n}
  \address{Universidad de Buenos Aires, Argentina}
  \author{Alejandro R\'{i}os}
  \address{Universidad de Buenos Aires, Argentina}
  \author{Andr\'{e}s Viso}
  \address{Universidad de Buenos Aires, Argentina \\
           Universidad Nacional de Quilmes, Argentina}
  \thanks[INFINIS]{This work was partially supported by LIA INFINIS, and the
                   ECOS-Sud program PA17C01.}

\begin{abstract}
It is well-known in the field of programming languages that dealing with
variable names and binders may lead to conflicts such as undesired captures
when implementing interpreters or compilers. This situation has been overcome
by resorting to de Bruijn indices for calculi where binders capture only one
variable name, like the $\calcLambda$-calculus. The advantage of this approach
relies on the fact that so-called $\alpha$-equivalence becomes syntactical
equality when working with indices.

In recent years pattern calculi have gained considerable attention given their
expressiveness. They turn out to be notoriously convenient to study the
foundations of modern functional programming languages modeling features like
pattern matching, path polymorphism, pattern polymorphism, etc. However, the
literature falls short when it comes to dealing with $\alpha$-conversion and
binders capturing simultaneously several variable names. Such is the case of
the \emph{Pure Pattern Calculus} ($\calcPPC$): a natural extension of
$\calcLambda$-calculus that allows to abstract virtually any term.

This paper extends de Bruijn's ideas to properly overcome the multi-binding
problem by introducing a novel presentation of $\calcPPC$ with bidimensional
indices, in an effort to implement a prototype for a typed functional
programming language based on $\calcPPC$ that captures path polymorphism.
\end{abstract}

\end{frontmatter}

%%%%%%%%%%%%%%%%%%%%%%%%%%%%%%%%%%%%%%%%%%%%%%%%%%%%%%%%%%%%%%%%%%%%%%%%%%%%%%%
\section{Introduction}
\label{c:intro}
%%%%%%%%%%%%%%%%%%%%%%%%%%%%%%%%%%%%%%%%%%%%%%%%%%%%%%%%%%%%%%%%%%%%%%%%%%%%%%%

The foundations of functional programming languages like LISP, Miranda, Haskell
or the ones in the ML family (Caml, SML, OCaml, etc.) rely strongly on the
study of the $\calcLambda$-calculus~\cite{Barendregt85} and its many variants
introduced over the years. Among them there are the \emph{pattern
calculi}~\cite{Oostrom90,CirsteaK98,Kahl03,CerritoK04,JayK06,Jay04,KlopOV08},
whose key feature can be identified as \emph{pattern-matching}.
Pattern-matching has been extensively used in programming languages as a means
for writing succinct and elegant programs. It stands for the possibility of
defining functions by cases, analysing the shape of their arguments, while
providing a syntactic tool to decompose such arguments in their parts when
applying the function.

In the standard $\calcLambda$-calculus, functions are represented by
expressions of the form $\termabs{x}{t}$, where $x$ is the formal parameter and
$t$ the body of the function. Such a function may be applied to any term,
regardless of its form, as dictated by the $\rbeta$-reduction rule:
$\termapp{(\termabs{x}{t})}{u} \rrule{\rbeta} \subsapply{\subs{x}{u}}{t}$,
where $\subsapply{\subs{x}{u}}{t}$ stands for the result of replacing all free
occurrences of $x$ in $t$ by $u$. Note that no requirement on the shape of $u$
is placed. Pattern calculi, on the contrary, provide generalisations of the
$\rbeta$-reduction rule in which abstractions $\termabs{x}{t}$ are replaced by
more general terms like $\termabs{p}{t}$ where $p$ is called a \emph{pattern}.
For example, consider the function $\termabs{\termpair{x}{y}}{x}$ that projects
the first component of a pair. Here the pattern is the pair $\termpair{x}{y}$
and the expression $\termapp{(\termabs{\termpair{x}{y}}{x})}{u}$ will only be
able to reduce if $u$ is indeed of the form $\termpair{u_1}{u_2}$. Otherwise,
reduction will be blocked.

We are particularly interested in studying the \emph{Pure Pattern
Calculus} ($\calcPPC$)~\cite{JayK09} and the novel features it introduced in
the field of pattern calculi, namely \emph{path polymorphism} and \emph{pattern
polymorphism}. The former refers to the possibility of defining functions that
uniformly traverse arbitrary data structures, while the latter allows to
consider patterns as parameters that may be dynamically generated in run-time.
Developing such a calculus implies numerous technical challenges to guarantee
well-behaved operational semantics in the untyped framework. Recently, a static
type system has been introduced for a restriction of $\calcPPC$ called
\emph{Calculus of Applicative Patterns} ($\calcCAP$)~\cite{VisoBA16}, which is
able to capture the path polymorphic aspect of $\calcPPC$. Moreover,
type-checking algorithms for such a formalism has also been
studied~\cite{EdiVB15}, as a first step towards an implementation of a
prototype for a typed functional programming language capturing such features.
Following this line of research, studies on the definition of normalising
strategies for $\calcPPC$ have been done as
well~\cite{BonelliKLR12,BonelliKLR17}. Such results are ported to $\calcCAP$ by
means of a simple embedding~\cite{Viso20} where the static typing discipline
gives further guarantees on the well-behaved semantics of terms.

Within this framework, the present work aims to throw some light on the
implementation aspects of these formalisms. In particular, working modulo
$\alpha$-conversion~\cite{Barendregt85} implies dealing with variable renaming
during the implementation. Such an approach is known to be error-prone and
computationally expensive. One way of getting rid of this problem in the
$\calcLambda$-calculus setting is adopting de Bruijn
notation~\cite{Bruijn72,Bruijn78}, a technique that simply avoids working
modulo $\alpha$-conversion. To the best of our knowledge, no dynamic pattern
calculi in the likes of $\calcPPC$ with de Bruijn indices has been formalised
in the literature. However, there are some references worth mentioning.
In~\cite{OostromR14} an alternative presentation of $\calcPPC$ is given in the
framework of \emph{Higher-Order Pattern Rewriting System}
(HRS)~\cite{Nipkow91,MayrN98}, together with translations between the two
systems. On the other hand, in~\cite{BonelliKR05} de Bruijn ideas had been
extended to \emph{Expression Reduction Systems} (ERS)~\cite{GlauertKK05} also
providing formal translations from systems with names to systems with indices,
and vice-versa. Moreover, the correspondence between HRS and ERS has already
been established~\cite{Raamsdonk80}. The composition of such translations might
derive a higher order system \emph{\`{a} la} de Bruijn capturing the features
of $\calcPPC$. However, this would result in a rather indirect solution to our
problem where many technicalities still need to be sorted out.

We aim to formalise an intuitive variant of $\calcPPC$ with de Bruijn indices
where known results for the original calculus, such as the existence of
normalising strategies, may easily be ported and reused.

%%%%%%%%%%%%%%%%%%%%%%%%%%%%%%%%%%%%%%%%%%%%%%%%%%%%%%%%%%%%%%%%%%%%%%%%%%%%%%%
\subsection{Contributions}
\label{c:intro:contributions}
%%%%%%%%%%%%%%%%%%%%%%%%%%%%%%%%%%%%%%%%%%%%%%%%%%%%%%%%%%%%%%%%%%%%%%%%%%%%%%%

This paper extends de Bruijn's ideas to handle binders that capture multiple
symbols at once, by means of what we call \emph{bidimensional indices}. These
ideas are illustrated by introducing a novel presentation of $\calcPPC$,
without variable/matchable names, called $\calcPPCX$. Moreover, binders in the
new proposed calculus are capable of handling two kinds of indices, namely
variable and matchable indices, as required by the $\calcPPC$ operational
semantics.

Proper translations from $\calcPPC$ to $\calcPPCX$ and back are introduced.
This functions preserve the matching operation and, hence, the operational
semantics of both calculi. Moreover, they turn out to be the inverse of each
other. This leads to a crucial strong bisimulation result between the two
calculi, which allows to import many known properties of $\calcPPC$ into
$\calcPPCX$, for instance confluence and the existence of normalising
strategies.

%%%%%%%%%%%%%%%%%%%%%%%%%%%%%%%%%%%%%%%%%%%%%%%%%%%%%%%%%%%%%%%%%%%%%%%%%%%%%%%
\subsection{Structure of the paper}
\label{c:intro:structure}
%%%%%%%%%%%%%%%%%%%%%%%%%%%%%%%%%%%%%%%%%%%%%%%%%%%%%%%%%%%%%%%%%%%%%%%%%%%%%%%

We start by briefly introducing $\calcPPC$ and reminding the mechanism of de
Bruijn indices for the $\calcLambda$-calculus in Sec.~\ref{s:preliminaries}.
The novel $\calcPPCX$ is formalised in Sec.~\ref{s:ppcx}, followed by the
introduction of the translations in Sec.~\ref{s:translation}. The strong
bisimulation result is presented in Sec.~\ref{s:bisimulation} together with a
discussion of different properties of $\calcPPCX$ that follow from it. We
conclude in Sec.~\ref{s:conclusion} and discuss possible lines of future work.

%%% Local Variables:
%%% mode: latex
%%% TeX-master: "main"
%%% End:

%%%%%%%%%%%%%%%%%%%%%%%%%%%%%%%%%%%%%%%%%%%%%%%%%%%%%%%%%%%%%%%%%%%%%%%%%%%%%%%
\section{Preliminaries}
\label{s:preliminaries}
%%%%%%%%%%%%%%%%%%%%%%%%%%%%%%%%%%%%%%%%%%%%%%%%%%%%%%%%%%%%%%%%%%%%%%%%%%%%%%%

This section introduces preliminary concepts that guide our development and
will help the reader follow the new ideas presented in this work.

%%%%%%%%%%%%%%%%%%%%%%%%%%%%%%%%%%%%%%%%%%%%%%%%%%%%%%%%%%%%%%%%%%%%%%%%%%%%%%%
\subsection{The Pure Pattern Calculus}
\label{s:preliminaries:ppc}
%%%%%%%%%%%%%%%%%%%%%%%%%%%%%%%%%%%%%%%%%%%%%%%%%%%%%%%%%%%%%%%%%%%%%%%%%%%%%%%

We start by briefly introducing the \emph{Pure Pattern Calculus}
($\calcPPC$)~\cite{JayK09}, an extension of the $\calcLambda$-calculus where
virtually any term can be abstracted. This gives place to two versatile forms
of polymorphism that set the foundations for adding novel features to future
functional programming languages: namely \emph{path polymorphism} and
\emph{pattern polymorphism}. This work, however, focuses on implementation
related aspects of $\calcPPC$ and will not delve deeper into these new forms of
polymorphism. We refer the reader to~\cite{JayK09,Jay09} for an in-depth study
of them.

\medskip
Given an infinitely countable set of \emph{symbols} $\TermVariable$ ($x, y, z,
\ldots$), the sets of \emph{terms} $\Term{\calcPPC}$ and \emph{contexts} are
given by the following grammar: \[
\begin{array}{c@{\qquad}c}
\begin{array}{rrcl}
\textbf{Terms}    & t         & \Coloneq  & \termvar{x} \mid \termmatch{x} \mid \termapp{t}{t} \mid \termabs[\theta]{t}{t}
\end{array}
&
\begin{array}{rrcl}
\textbf{Contexts} & \ctxt{C}  & \Coloneq  & \Box \mid \termapp{\ctxt{C}}{t} \mid \termapp{t}{\ctxt{C}} \mid \termabs[\theta]{\ctxt{C}}{t} \mid \termabs[\theta]{t}{\ctxt{C}}
\end{array}
\end{array}
\] where $\theta$ is a list of symbols that are bound by the abstraction. A
symbol $\termvar{x}$ appearing in a term is dubbed a \emph{variable symbol}
while $\termmatch{x}$ is called a \emph{matchable symbol}. In particular,
given $\termabs[\theta]{p}{t}$, $\theta$ binds variable symbols in the
\emph{body} $t$ and matchable symbols in the \emph{pattern} $p$. Thus, the set
of \emph{free variables} and \emph{free matchables} of a term $t$, written
$\fv{t}$ and $\fm{t}$ respectively, are inductively defined as: \[
\begin{array}{c@{\qquad}c}
\begin{array}{rcl}
\fv{\termvar{x}}            & \eqdef  & \set{x} \\
\fv{\termmatch{x}}          & \eqdef  & \emptyset \\
\fv{\termapp{t}{u}}         & \eqdef  & \fv{t} \cup \fv{u} \\
\fv{\termabs[\theta]{p}{t}} & \eqdef  & \fv{p} \cup (\fv{t} \setminus \theta)
\end{array}
&
\begin{array}{rcl}
\fm{\termvar{x}}            & \eqdef  & \emptyset \\
\fm{\termmatch{x}}          & \eqdef  & \set{x} \\
\fm{\termapp{t}{u}}         & \eqdef  & \fm{t} \cup \fm{u} \\
\fm{\termabs[\theta]{p}{t}} & \eqdef  & (\fm{p} \setminus \theta) \cup \fm{t}
\end{array}
\end{array}
\] A term is said to be \emph{closed} if it has no free variables. Note that
free matchables are allowed, and should be understood as \emph{constants} or
\emph{constructors} for data structures. The pattern $p$ of an abstraction
$\termabs[\theta]{p}{t}$ is \emph{linear} if every symbol $x \in \theta$ occurs
at most once in $p$.

\begin{wrapfigure}{R}{0.2\textwidth}
\begin{tikzpicture}[baseline=(current bounding box.center),->,auto,node distance=5pt]
\node (ABS1a)                               {$\vphantom{(}\itermabs_{\lista{x}}$};
\node (PTN1)  [right of=ABS1a, xshift=6pt]  {$\vphantom{(}\termmatch{x}$};
\node (ABS1b) [right of=PTN1,  xshift=1pt]  {$\vphantom{(}.($};
\node (ABS2a) [right of=ABS1b, xshift=6pt]  {$\vphantom{(}\itermabs_{\lista{y}}$};
\node (PTN2a) [right of=ABS2a, xshift=6pt]  {$\vphantom{(}\termvar{x}$};
\node (PTN2b) [right of=PTN2a, xshift=3pt]  {$\vphantom{(}\termmatch{y}$};
\node (ABS2b) [right of=PTN2b]              {$\vphantom{(}.$};
\node (BODY)  [right of=ABS2b]              {$\vphantom{(}\termvar{y}$};
\node (ABS1c) [right of=BODY]               {$\vphantom{(})$};

\path (ABS1a.south) edge [bend right=45] node[left] {} (PTN1.south)
      (ABS1a.south) edge [bend right=90] node[left] {} (PTN2a.south)
      (ABS2a.south) edge [bend right=90] node[left] {} (PTN2b.south)
      (ABS2a.south) edge [bend right=90] node[left] {} (BODY.south);
\end{tikzpicture}
\end{wrapfigure}

To illustrate how variables and matchables are bound, consider the function
$\mathsf{elim}$ defined as
$\termabs[\lista{x}]{\termmatch{x}}{(\termabs[\lista{y}]{\termapp{\termvar{x}}{\termmatch{y}}}{\termvar{y}})}$.
The inner abstraction binds the only occurrence of the matchable
$\termmatch{y}$ in the pattern $\termapp{\termvar{x}}{\termmatch{y}}$ and that
of the variable $\termvar{y}$ in the body $\termvar{y}$. However, the
occurrence of $\termvar{x}$ in $\termapp{\termvar{x}}{\termmatch{y}}$ is not
bound by the inner abstraction, as it is excluded from $\lista{y}$, acting as a
place-holder in that pattern. It is the outermost abstraction that binds both
$\termvar{x}$ in the inner pattern and $\termmatch{x}$ in the outermost
pattern. This is graphically depicted above.

\medskip
A \emph{substitution} ($\sigma, \rho, \ldots$) is a partial function from
variables to terms. The substitution $\sigma = \subs{x_i}{u_i}_{i \in I}$,
where $I$ is a set of indices, maps the variable $x_i$ into the term $u_i$ (\ie
$\funcapply{\sigma}{x_i} \eqdef u_i$) for each $i \in I$. Thus, its
\emph{domain} and \emph{image} are defined as $\dom{\sigma} \eqdef
\set{x_i}_{i \in I}$ and $\img{\sigma} \eqdef \set{u_i}_{i \in I}$
respectively. For convenience, a substitution $\sigma$ is usually turned into a
total function by defining $\funcapply{\sigma}{x} \eqdef x$ for every $x \notin
\dom{\sigma}$. Then, the identity substitution is denoted $\subsid$ or simply
$\funcid$.

A \emph{match} ($\mu, \nu, \ldots$) may be successful (yielding a
substitution), it may fail (returning a special symbol $\matchfail$) or be
undetermined (denoted by a special symbol $\matchundet$). The cases of success
and failure are called \emph{decided matches}. All concepts and notation
relative to substitutions are extended to matches so that, for example, the
domain of $\matchfail$ is empty while that of $\matchundet$ is undefined. The
sets of free variable and free matchable symbols of $\sigma$ are defined as the
union of $\fv{\subsapply{\sigma}{x}}$ and $\fm{\subsapply{\sigma}{x}}$ for
every $x \in \dom{\sigma}$ respectively, while $\fv{\matchfail} =
\fm{\matchfail} = \emptyset$ and they are undefined for $\matchundet$. The set
of symbols of a substitution is defined as $\sym{\sigma} \eqdef \dom{\sigma}
\cup \fv{\sigma} \cup \fm{\sigma}$. The predicate $\pavoids{x}{\sigma}$ states
that $x \notin \sym{\sigma}$. It is extended to sets and matches as expected.
In particular, $\pavoids{\theta}{\mu}$ implies that $\mu$ must be decided.

The result of applying a substitution $\sigma$ to a term $t$, denoted
$\subsapply{\sigma}{t}$, is inductively defined as: \[
\begin{array}{c@{\qquad}c}
\begin{array}{rcl@{\quad}l}
\subsapply{\sigma}{\termvar{x}}             & \eqdef  & \funcapply{\sigma}{x} \\
\subsapply{\sigma}{\termmatch{x}}           & \eqdef  & \termmatch{x}
\end{array}
&
\begin{array}{rcl@{\quad}l}
\subsapply{\sigma}{(\termapp{t}{u})}        & \eqdef  & \termapp{\subsapply{\sigma}{t}}{\subsapply{\sigma}{u}} \\
\subsapply{\sigma}{\termabs[\theta]{p}{t}}  & \eqdef  & \termabs[\theta]{\subsapply{\sigma}{p}}{\subsapply{\sigma}{t}}  & \text{if $\pavoids{\theta}{\sigma}$}
\end{array}
\end{array} \] The restriction in the case of the abstraction is required to
avoid undesired captures of variables/matchables. However, it can always be
satisfied by resorting to $\alpha$-conversion.

The result of applying a match $\mu$ to a term $t$, denoted
$\matchapply{\mu}{t}$, is defined as:
\begin{inparaenum}
  \item if $\mu = \sigma$ a substitution, then $\matchapply{\mu}{t} \eqdef
  \subsapply{\sigma}{t}$;
  \item if $\mu = \matchfail$, then $\matchapply{\mu}{t} \eqdef
  \termabs[\lista{x}]{\termmatch{x}}{\termvar{x}}$ (\ie the identity function);
  or
  \item if $\mu = \matchundet$, then $\matchapply{\mu}{t}$ is undefined.
\end{inparaenum}

The \emph{composition} $\subscomp{\sigma}{\sigma'}$ of substitutions is defined
as usual, \ie $\subsapply{(\subscomp{\sigma}{\sigma'})}{x} \eqdef
\subsapply{\sigma}{(\subsapply{\sigma'}{x})}$, and the notion is extended to
matches by defining $\matchcomp{\mu}{\mu'} \eqdef \matchfail$ if any of the two
matches is $\matchfail$. Otherwise, if at least one of the two is
$\matchundet$, then $\matchcomp{\mu}{\mu'} \eqdef \matchundet$. In particular,
$\matchcomp{\matchfail}{\matchundet} = \matchfail$. The \emph{disjoint union}
$\matchunion{\mu}{\mu'}$ of matches is defined as follows:
\begin{inparaenum}
  \item if $\mu = \matchfail$ or $\mu' = \matchfail$, then
  $\matchunion{\mu}{\mu'} \eqdef \matchfail$; else
  \item if $\mu = \matchundet$ or $\mu' = \matchundet$, then
  $\matchunion{\mu}{\mu'} \eqdef \matchundet$; otherwise
  \item both $\mu$ and $\mu'$ are substitutions and if $\dom{\mu} \cap
  \dom{\mu'} \neq \emptyset$, then $\matchunion{\mu}{\mu'} \eqdef \matchfail$,
  else: \[\subsapply{(\matchunion{\mu}{\mu'})}{x} \eqdef
  \begin{cases}
  \subsapply{\mu}{x}  & \text{if $x \in \dom{\mu}$} \\
  \subsapply{\mu'}{x} & \text{if $x \in \dom{\mu'}$} \\
  x                   & \text{otherwise}
  \end{cases}\]
\end{inparaenum}
Disjoint union is used to guarantee that the matching operation is
deterministic.

\medskip
Before introducing the matching operation it is necessary to motivate the
concept of \emph{matchable form}. The pattern
$\termapp{\termmatch{x}}{\termmatch{y}}$ allows, at first, to decompose
arbitrary applications, which may lead to the loss of confluence. For
instance: \[
\begin{array}{llllll}
\termapp{(\termabs[\lista{x,y}]{\termapp{\termmatch{x}}{\termmatch{y}}}{\termvar{y}})}{(\termapp{(\termabs[\lista{w}]{\termmatch{w}}{\termapp{\termmatch{z_0}}{\termmatch{z_1}}})}{\termmatch{z_0}})} & \reduce & \termapp{(\termabs[\lista{x,y}]{\termapp{\termmatch{x}}{\termmatch{y}}}{\termvar{y}})}{(\termapp{\termmatch{z_0}}{\termmatch{z_1}})} & \reduce & \termmatch{z_1} \\
\termapp{(\termabs[\lista{x,y}]{\termapp{\termmatch{x}}{\termmatch{y}}}{\termvar{y}})}{(\termapp{(\termabs[\lista{w}]{\termmatch{w}}{\termapp{\termmatch{z_0}}{\termmatch{z_1}}})}{\termmatch{z_0}})} & \reduce & \termmatch{z_0}
\end{array}
\] This issue arises when allowing to match the pattern
$\termapp{\termmatch{x}}{\termmatch{y}}$ with an application that may still be
reduced, like the argument
$\termapp{(\termabs[\lista{w}]{\termmatch{w}}{\termapp{\termmatch{z_0}}{\termmatch{z_1}}})}{\termmatch{z_0}}$
of the outermost redex in the example above. To avoid this situation it is
required for the match to be decided only if the argument is sufficiently
evaluated. An analogous issue occurs if the pattern is reducible. Thus, both
the pattern and the argument must be in matchable form for the match to be
decided. The set of \emph{data structures} $\TermData{\calcPPC}$ and
\emph{matchable forms} $\Matchable{\calcPPC}$ are given by the following
grammar: \[
\begin{array}{c@{\qquad}c}
\begin{array}{rrcl}
\textbf{Data structures}  & d & \Coloneq  & \termmatch{x} \mid \termapp{d}{t}
\end{array}
&
\begin{array}{rrcl}
\textbf{Matchable forms}  & m & \Coloneq  & d \mid \termabs[\theta]{t}{t}
\end{array}
\end{array} \]

The \emph{matching operation} $\match[\theta]{p}{u}$ of a pattern $p$ against a
term $u$ relative to a list of symbols $\theta$ is defined as the application,
in order, of the following equations: \[
\begin{array}{rcl@{\qquad}l}
\match[\theta]{\termmatch{x}}{u}                & \eqdef  & \subs{\termvar{x}}{u}                                   & \text{if $x \in \theta$} \\
\match[\theta]{\termmatch{x}}{\termmatch{x}}    & \eqdef  & \subsid                                                 & \text{if $x \notin \theta$} \\
\match[\theta]{\termapp{p}{q}}{\termapp{t}{u}}  & \eqdef  & \matchunion{\match[\theta]{p}{t}}{\match[\theta]{q}{u}} & \text{if $\termapp{t}{u}, \termapp{p}{q} \in \Matchable{\calcPPC}$} \\
\match[\theta]{p}{u}                            & \eqdef  & \matchfail                                              & \text{if $u, p \in \Matchable{\calcPPC}$} \\
\match[\theta]{p}{u}                            & \eqdef  & \matchundet                                             & \text{otherwise}
\end{array}
\] An additional check is imposed, namely $\dom{\match[\theta]{p}{u}} =
\theta$. Otherwise, $\match[\theta]{p}{u} \eqdef \matchfail$. This last
condition is necessary to prevent bound symbols from going out of scope when
reducing. It can be easily guaranteed though by requesting, for each
abstraction $\termabs[\theta]{p}{t}$, that $\theta \subseteq \fm{p}$.
For instance, consider the term
$\termapp{(\termabs[\lista{x,y}]{\termmatch{x}}{\termvar{y}})}{u}$. Without
this final check, matching the argument $u$ against the pattern
$\termmatch{x}$ would yield a substitution $\subs{\termvar{x}}{u}$ and no
term would be assigned to the variable $\termvar{y}$ in the body of the
abstraction.

\medskip
Finally, the reduction relation $\reduce[\rPPC]$ of $\calcPPC$ is given by the
closure by contexts of the rewriting rule: \[
\termapp{(\termabs[\theta]{p}{s})}{u} \rrule{\rPPC} \matchapply{\match[\theta]{p}{u}}{s}
\] whenever $\match[\theta]{p}{u}$ is a decided match. To illustrate the
operational semantics of $\calcPPC$ consider the term $\mathsf{elim}$
introduced above, applied to the function
$\termabs[\lista{z}]{\termmatch{z}}{\termapp{\termapp{\termmatch{c}}{\termvar{z}}}{\termmatch{n}}}$
where the free matchables $\termmatch{c}$ and $\termmatch{n}$ can be seen as
constructors for lists $\mathsf{cons}$ and $\mathsf{nil}$ respectively: \[
\termapp{(\termabs[\lista{x}]{\termmatch{x}}{(\termabs[\lista{y}]{\termapp{\termvar{x}}{\termmatch{y}}}{\termvar{y}})})}{(\termabs[\lista{z}]{\termmatch{z}}{\termapp{\termapp{\termmatch{c}}{\termvar{z}}}{\termmatch{n}}})}
\quad\reduce[\rPPC]\quad
\termabs[\lista{y}]{\termapp{(\termabs[\lista{z}]{\termmatch{z}}{\termapp{\termapp{\termmatch{c}}{\termvar{z}}}{\termmatch{n}}})}{\termmatch{y}}}{\termvar{y}}
\quad\reduce[\rPPC]\quad
\termabs[\lista{y}]{\termapp{\termapp{\termmatch{c}}{\termmatch{y}}}{\termmatch{n}}}{\termvar{y}}
\] In the first step,
$\termabs[\lista{z}]{\termmatch{z}}{\termapp{\termapp{\termmatch{c}}{\termvar{z}}}{\termmatch{n}}}$
is substituted for $\termvar{x}$ into the pattern
$\termapp{\termvar{x}}{\termmatch{y}}$. In the second step, the resulting
application, which resides in the pattern, is reduced. The resulting term, when
applied to an argument, will yield a successful matching only if this argument
is a compound data of the form
$\termapp{\termapp{\termmatch{c}}{t}}{\termmatch{n}}$.

This relation is shown to be \emph{confluent} (CR) based on the matching
operation introduced above.

\begin{theorem}[\cite{JayK09}]
The reduction relation $\reduce[\rPPC]$ is confluent (CR).
\label{t:preliminaries:ppc:cr}
\end{theorem}

%%%%%%%%%%%%%%%%%%%%%%%%%%%%%%%%%%%%%%%%%%%%%%%%%%%%%%%%%%%%%%%%%%%%%%%%%%%%%%%
\subsection{de Bruijn indices}
\label{s:preliminaries:indices}
%%%%%%%%%%%%%%%%%%%%%%%%%%%%%%%%%%%%%%%%%%%%%%%%%%%%%%%%%%%%%%%%%%%%%%%%%%%%%%%

We introduce next de Bruijn indices for the $\calcLambda$-calculus. Among the
many presentations of de Bruijn indices in the literature, we will follow that
of~\cite{KamareddineR95} as our development builds upon their ideas. In
particular, we choose to work with the presentation where indices are partially
updated as the term is being traversed by the substitution operation (details
below). We refer the reader to~\cite{KamareddineR95} for the equivalent version
where the update is performed once at the end of the substitution process. We
introduce now the \emph{$\calcLambda$-calculus with de Bruijn indices}
($\calcLambdaX$ for short).

\medskip
The sets of \emph{terms} $\Term{\calcLambdaX}$ and \emph{contexts} are given by
the following grammar: \[
\begin{array}{c@{\qquad}c}
\begin{array}{rrcl}
\textbf{Terms}    & t         & \Coloneq  & \termidxv{i} \mid \termapp{t}{t} \mid \termabsi{t} \\
\end{array}
&
\begin{array}{rrcl}
\textbf{Contexts} & \ctxt{C}  & \Coloneq  & \Box \mid \termapp{\ctxt{C}}{t} \mid \termapp{t}{\ctxt{C}} \mid \termabsi{\ctxt{C}}
\end{array}
\end{array}
\] where $\termidxv{i} \in \Natural_{\geq 1}$ is called an \emph{index}.
Indices are place-holders indicating the distance to the binding abstraction.
In the context of the $\calcLambdaX$-calculus, indices are also called
\emph{variables}. Thus, the \emph{free variables} of a term are inductively
defined as: $\fv{\termidxv{i}} \eqdef \set{\termidx{i}}$; $\fv{\termapp{t}{u}}
\eqdef \fv{t} \cup \fv{u}$; and $\fv{\termabsi{t}} \eqdef \fv{t} - 1$, where
$X - k$ stands for subtracting $k$ from each element of the set $X$, removing
those that result in a non-positive index.

In order to define $\rbeta$-reduction \emph{\`{a} la} de Bruijn, the
substitution of an index $\termidxv{i}$ for a term $u$ in a term $t$ must be
defined. Therefore, it is necessary to identify among the indices of the term
$t$, those corresponding to $\termidx{i}$. Furthermore, the indices of $u$
should be updated in order to preserve the correct bindings after the
replacement of the variable by $u$. To that end, the \emph{increment at depth
$k$} for variables in a term $t$, written $\fincv[k]{t}$, is inductively
defined as follows: \[
\begin{array}{c@{\qquad\qquad}c}
\begin{array}{rcll}
\fincv[k]{\termidxv{i}}   & \eqdef  & \begin{cases}
                                      \termidxv{i+1}  & \text{if $i > k$} \\
                                      \termidxv{i}    & \text{if $i \leq k$}
                                      \end{cases}
\end{array}
&
\begin{array}{rcll}
\fincv[k]{\termapp{t}{u}} & \eqdef  & \termapp{\fincv[k]{t}}{\fincv[k]{u}} \\
\fincv[k]{\termabsi{t}}   & \eqdef  & \termabsi{\fincv[k+1]{t}}
\end{array}
\end{array}
\] Then, the \emph{substitution at level $i$} of a term $u$ in a term $t$,
denoted $\subsapply{\subs{\termidx{i}}{u}}{t}$, is defined as a partial
function mapping free variables at level $i$ to terms, performing the
appropriate updates as it traverses the substituted term, to avoid undesired
captures. \[
\begin{array}{c@{\qquad\qquad}c}
\begin{array}{rcl@{\quad}l}
\subsapply{\subs{\termidxv{i}}{u}}{\termidxv{i'}}     & \eqdef  & \begin{cases}
                                                                  \termidxv{i'-1} & \text{if $i' > i$} \\
                                                                  u               & \text{if $i' = i$} \\
                                                                  \termidxv{i'}   & \text{if $i' < i$}
                                                                  \end{cases}
\end{array}
&
\begin{array}{rcl@{\quad}l}
\subsapply{\subs{\termidxv{i}}{u}}{(\termapp{t}{s})}  & \eqdef  & \termapp{\subsapply{\subs{\termidxv{i}}{u}}{t}}{\subsapply{\subs{\termidxv{i}}{u}}{s}} \\
\subsapply{\subs{\termidxv{i}}{u}}{\termabsi{t}}      & \eqdef  & \termabsi{\subsapply{\subs{\termidxv{i+1}}{\fincv[0]{u}}}{t}}
\end{array}
\end{array} \] It is worth noticing that this substitution should be
interpreted in the context of a redex, where a binder is removed and its bound
index substituted. This forces to update the free indices, that might be
captured by an outermost abstraction, as done by the first case of the
substitution over a variable $\termidxv{i'}$. Hence, preserving the correct
bindings.

Finally, the reduction relation $\reduce[\rbetax]$ of the
$\calcLambdaX$-calculus is given by the closure by contexts of the rewriting
rule: \[
\termapp{(\termabsi{s})}{u} \rrule{\rbetax} \subsapply{\subs{\termidxv{1}}{u}}{s}
\]

Also, embeddings between the $\calcLambda$-calculus and $\calcLambdaX$ are
defined: $\toLambdaX{\_} : \Term{\calcLambda} \to \Term{\calcLambdaX}$ and
$\toLambda{\_} : \Term{\calcLambdaX} \to \Term{\calcLambda}$, in such a way
that they are the inverse of each other and they allow to simulate one
calculus into the other:

\begin{theorem}[\cite{KamareddineR95}]
Let $t \in \Term{\calcLambda}$ and $s \in \Term{\calcLambdaX}$. Then,
\begin{enumerate}
  \item \label{t:preliminaries:indices:simulation:lambdax} If $t
  \reduce[\rbeta] t'$, then $\toLambdaX{t} \reduce[\rbetax] \toLambdaX{t'}$.

  \item \label{t:preliminaries:indices:simulation:lambda} If $s
  \reduce[\rbetax] s'$, then $\toLambda{s} \reduce[\rbeta] \toLambda{s'}$.
\end{enumerate}
\label{t:preliminaries:indices:simulation}
\end{theorem}

This shows that both formalisms ($\calcLambda$-calculus and $\calcLambdaX$)
have exactly the same operational semantics.

As an example to illustrate both reduction in the $\calcLambdaX$-calculus and
its equivalence with the $\calcLambda$-calculus, consider the following terms:
$\termapp{\termapp{(\termabs{z}{\termabs{y}{z}})}{(\termabs{x}{x})}}{(\termabs{x}{\termapp{x}{x}})}$
and
$\termapp{\termapp{(\termabsi{\termabsi{\termidxv{2}}})}{(\termabsi{\termidxv{1}})}}{(\termabsi{\termapp{\termidxv{1}}{\termidxv{1}}})}$.
The reader can verify that both expressions encode the same function in its
respective calculus. As expected, their operational semantics coincide \[
\begin{array}{c@{\quad}c@{\quad}c@{\quad}c@{\quad}c}
\termapp{\termapp{(\termabs{z}{\termabs{y}{z}})}{(\termabs{x}{x})}}{(\termabs{x}{\termapp{x}{x}})}
& \reduce[\rbeta] &
\termapp{(\termabs{y}{\termabs{x}{x}})}{(\termabs{x}{\termapp{x}{x}})}
& \reduce[\rbeta] &
\termabs{x}{x} \\
\termapp{\termapp{(\termabsi{\termabsi{\termidxv{2}}})}{(\termabsi{\termidxv{1}})}}{(\termabsi{\termapp{\termidxv{1}}{\termidxv{1}}})}
& \reduce[\rbetax] &
\termapp{(\termabsi{\termabsi{\termidxv{1}}})}{(\termabsi{\termapp{\termidxv{1}}{\termidxv{1}}})}
& \reduce[\rbetax] &
\termabsi{\termidxv{1}} \\
\end{array} \]

%%% Local Variables:
%%% mode: latex
%%% TeX-master: "main"
%%% End:

%%%%%%%%%%%%%%%%%%%%%%%%%%%%%%%%%%%%%%%%%%%%%%%%%%%%%%%%%%%%%%%%%%%%%%%%%%%%%%%
\section{The Pure Pattern Calculus with de Bruijn indices}
\label{s:ppcx}
%%%%%%%%%%%%%%%%%%%%%%%%%%%%%%%%%%%%%%%%%%%%%%%%%%%%%%%%%%%%%%%%%%%%%%%%%%%%%%%

This section introduces the novel \emph{Pure Pattern Calculus with de Bruijn
indices} ($\calcPPCX$). It represents a natural extension of de Bruijn ideas to
a framework where a binder may capture more than one symbol. In the particular
case of $\calcPPC$ there are two kinds of captured symbols, namely variables
and matchables. This distinction is preserved in $\calcPPCX$ while extending
indices to pairs (\emph{a.k.a.}~bidimensional indices) to distinguish the
binder that captures the symbol and the individual symbol among all those
captured by the same binder.

The sets of \emph{terms} $\Term{\calcPPCX}$, \emph{contexts}, \emph{data
structures} $\TermData{\calcPPCX}$ and \emph{matchable forms}
$\Matchable{\calcPPCX}$ of $\calcPPCX$ are given by the following grammar: \[
\begin{array}{c@{\qquad}c}
\begin{array}{rrcl}
\textbf{Terms}    & t         & \Coloneq  & \termidxv{i}[j] \mid \termidxm{i}[j] \mid \termapp{t}{t} \mid \termabs[n]{t}{t} \\
\textbf{Contexts} & \ctxt{C}  & \Coloneq  & \Box \mid \termapp{\ctxt{C}}{t} \mid \termapp{t}{\ctxt{C}} \mid \termabs[n]{\ctxt{C}}{t} \mid \termabs[n]{t}{\ctxt{C}}
\end{array}
&
\begin{array}{rrcl}
\textbf{Data structures}  & d & \Coloneq  & \termidxm{i}[j] \mid \termapp{d}{t} \\
\textbf{Matchable forms}  & m & \Coloneq  & d \mid \termabs[n]{t}{t}
\end{array}
\end{array}
\] where $\termidxv{i}[j]$ is dubbed a \emph{bidimensional index} and denotes
an ordered pair in $\Natural_{\geq 1} \times \Natural_{\geq 1}$ with
\emph{primary index} $\termidx{i}$ and \emph{secondary index} $\termidx{j}$.
The sub-index $n \in \Natural$ in an abstraction represents the amount of
indices (pairs) being captured by it. The primary index of a pair is used to
determine if the pair is bound by an abstraction, while the secondary index
identifies the pair among those (possibly many) bound ones. As for $\calcPPC$,
an index of the form $\termidxv{i}[j]$ is called a \emph{variable index} while
$\termidxm{i}[j]$ is dubbed a \emph{matchable index}. The \emph{free variables}
and \emph{free matchables} of a term are thus defined as follows: \[
\begin{array}{c@{\qquad}c}
\begin{array}{rcl}
\fv{\termidxv{i}[j]}    & \eqdef  & \set{\termidx{i}[j]} \\
\fv{\termidxm{i}[j]}    & \eqdef  & \emptyset \\
\fv{\termapp{t}{u}}     & \eqdef  & \fv{t} \cup \fv{u} \\
\fv{\termabs[n]{p}{t}}  & \eqdef  & \fv{p} \cup (\fv{t} - 1)
\end{array}
&
\begin{array}{rcl}
\fm{\termidxv{i}[j]}    & \eqdef  & \emptyset \\
\fm{\termidxm{i}[j]}    & \eqdef  & \set{\termidx{i}[j]} \\
\fm{\termapp{t}{u}}     & \eqdef  & \fm{t} \cup \fm{u} \\
\fm{\termabs[n]{p}{t}}  & \eqdef  & (\fm{p} - 1) \cup \fm{t}
\end{array}
\end{array}
\] where $X - k$ stands for subtracting $k$ from the primary index of each
element of the set $X$, removing those that result in a non-positive index.

\begin{wrapfigure}{R}{0.2\textwidth}
\begin{tikzpicture}[baseline=(current bounding box.center),->,auto,node distance=5pt]
\node (ABS1a)                               {$\vphantom{(}\itermabs_{1}$};
\node (PTN1)  [right of=ABS1a, xshift=5pt]  {$\vphantom{(}\termidxm{1}[1]$};
\node (ABS1b) [right of=PTN1,  xshift=3pt]  {$\vphantom{(}.($};
\node (ABS2a) [right of=ABS1b, xshift=4pt]  {$\vphantom{(}\itermabs_{1}$};
\node (PTN2a) [right of=ABS2a, xshift=5pt]  {$\vphantom{(}\termidxv{1}[1]$};
\node (PTN2b) [right of=PTN2a, xshift=6pt]  {$\vphantom{(}\termidxm{1}[1]$};
\node (ABS2b) [right of=PTN2b, xshift=1pt]  {$\vphantom{(}.$};
\node (BODY)  [right of=ABS2b, xshift=1pt]  {$\vphantom{(}\termidxv{1}[1]$};
\node (ABS1c) [right of=BODY]               {$\vphantom{(})$};

\path (ABS1a.south) edge [bend right=45] node[left] {} (PTN1.south)
      (ABS1a.south) edge [bend right=90] node[left] {} (PTN2a.south)
      (ABS2a.south) edge [bend right=90] node[left] {} (PTN2b.south)
      (ABS2a.south) edge [bend right=90] node[left] {} (BODY.south);
\end{tikzpicture}
\end{wrapfigure}

Let us illustrate these concepts with a similar example as that given for
$\calcPPC$, namely the function $\mathsf{elim} =
\termabs[\lista{x}]{\termmatch{x}}{(\termabs[\lista{y}]{\termapp{\termvar{x}}{\termmatch{y}}}{\termvar{y}})}$.
An equivalent term in the $\calcPPCX$ framework would be
$\termabs[1]{\termidxm{1}[1]}{(\termabs[1]{\termapp{\termidxv{1}[1]}{\termidxm{1}[1]}}{\termidxv{1}[1]})}$.
Note that variable indices in the context of a pattern are not bound by the
respective abstraction, in the same way that matchable indices in the body of
the abstraction are not captured either. Thus, the first occurrence of
$\termidxv{1}[1]$ is actually bound by the outermost abstraction, together with
the first occurrence of the matchable index $\termidxm{1}[1]$. The rest of the
indices in the term are bound by the inner abstraction as depicted in the
figure to the right. As a further (more interesting) example, consider the term
$\termabs[\lista{x,y}]{\termapp{\termmatch{x}}{\termmatch{y}}}{\termabs[\lista{}]{\termvar{x}}{\termvar{y}}}$
from $\calcPPC$, whose counter-part in $\calcPPCX$ would look like
$\termabs[2]{\termapp{\termidxm{1}[1]}{\termidxm{1}[2]}}{\termabs[0]{\termidxv{1}[1]}{\termidxv{2}[2]}}$.
This example illustrates the use of secondary indices to identify symbols
bound by the same abstraction. It also shows how the primary index of a
variable is increased when occurring within the body of an internal
abstraction, while this is not the case for occurrences in a pattern position.
Thus, both $\termidxv{1}[1]$ and $\termidxv{2}[2]$ are bound by the outermost
abstraction, as well as $\termidxm{1}[1]$ and $\termidxm{1}[2]$. Note how the
inner abstraction does not bind any index at all.

A term $t$ is said to be \emph{well-formed} if all the free bidimensional
indices (variables and matchables) of $t$ have their secondary index equal to
1, and for every sub-term of the form $\termabs[n]{p}{s}$ (written
$\termabs[n]{p}{s} \subterm t$) all the pairs captured by the abstraction have
their secondary index within the range $[1,n]$. Formally, $\set{\termidx{i}[j]
\mid \termidx{i}[j] \in \fm{t} \cup \fv{t}, j > 1} \cup
(\bigcup_{\termabs[n]{p}{s} \subterm t}{\set{\termidx{1}[j] \mid \termidx{1}[j] \in \fm{p} \cup \fv{s}, j > n}})
= \emptyset$.

\medskip
Before introducing a proper notion of substitution for $\calcPPCX$ it is
necessary to have a mechanism to update indices at a certain depth within the
term. The \emph{increment at depth $k$} for variable and matchable indices in a
term $t$, written $\fincv[k]{t}$ and $\fincm[k]{t}$ respectively, are
inductively defined as follows \[
\begin{array}{c@{\qquad}c}
\begin{array}{rcll}
\fincv[k]{\termidxv{i}[j]}    & \eqdef  & \begin{cases}
                                          \termidxv{(i+1)}[j] & \text{if $i > k$} \\
                                          \termidxv{i}[j]     & \text{if $i \leq k$}
                                          \end{cases} \\
\fincv[k]{\termidxm{i}[j]}    & \eqdef  & \termidxm{i}[j] \\
\fincv[k]{\termapp{t}{u}}     & \eqdef  & \termapp{\fincv[k]{t}}{\fincv[k]{u}} \\
\fincv[k]{\termabs[n]{p}{t}}  & \eqdef  & \termabs[n]{\fincv[k]{p}}{\fincv[k+1]{t}}
\end{array}
&
\begin{array}{rcll}
\fincm[k]{\termidxv{i}[j]}    & \eqdef  & \termidxv{i}[j] \\
\fincm[k]{\termidxm{i}[j]}    & \eqdef  & \begin{cases}
                                          \termidxm{(i+1)}[j] & \text{if $i > k$} \\
                                          \termidxm{i}[j]     & \text{if $i \leq k$}
                                          \end{cases} \\
\fincm[k]{\termapp{t}{u}}     & \eqdef  & \termapp{\fincm[k]{t}}{\fincm[k]{u}} \\
\fincm[k]{\termabs[n]{p}{t}}  & \eqdef  & \termabs[n]{\fincm[k+1]{p}}{\fincm[k]{t}}
\end{array}
\end{array}
\] Similarly, the \emph{decrement at depth} $k$ for variables ($\fdecv[k]{\_}$)
and matchables ($\fdecm[k]{\_}$) are defined by subtracting one from the
primary index above $k$ in the term. Most of the times these functions are used
with $k = 0$, thus the subindex will be omitted when it is clear from context.
In particular, the decrement function for variables will allow us to generalise
the idea of substitution at level $i$ with respect to the original one
presented in Sec.~\ref{s:preliminaries:indices}, which only holds in the
context of a $\rbeta$-reduction, by making the necessary adjustments to the
indices at the moment of the redution instead of hard-coding them into the
substitution meta-operation.

\begin{lemma}
Let $t \in \Term{\calcPPCX}$. Then,
\begin{enumerate}
  \item \label{l:appendix:finc:var} $\fincv[k]{\fincv[l]{t}} =
  \fincv[l]{\fincv[k-1]{t}}$ if $l < k$.
  \item \label{l:appendix:finc:match} $\fincv[k]{\fincm[l]{t}} =
  \fincm[l]{\fincv[k]{t}}$.
  \item \label{l:appendix:finc:fdec-l} $\fdecv[k]{\fincv[k]{t}} = t$.
  \item \label{l:appendix:finc:fdec-r} $\fincv[k]{\fdecv[k]{t}} = t$ iff
  $\termidx{(k+1)}[j] \notin \fv{t}$ for any $j \in \Natural_{\geq 1}$.
\end{enumerate}
\label{l:appendix:finc}
\end{lemma}

\begin{proof}
All items follow by straightforward induction on $t$.
\end{proof}

A \emph{substitution at level $i$} is a partial function from variable indices
to terms. It maps free variable indices at level $i$ to terms, performing the
appropriate updates as it traverses the substituted term, to avoid undesired
captures. \[
\begin{array}{c@{\quad}c}
\begin{array}{rcl@{\quad}l}
\subsapply{\subs{\termidxv{i}[j]}{u_j}_{j \in J}}{\termidxv{i'}[k]}   & \eqdef  & \begin{cases}
                                                                                u_k               & \text{if $i' = i, k \in J$} \\
                                                                                \termidxv{i'}[k]  & \text{if $i' \neq i$}
                                                                                \end{cases} \\
\subsapply{\subs{\termidxv{i}[j]}{u_j}_{j \in J}}{\termidxm{i'}[k]}   & \eqdef  & \termidxm{i'}[k]
\end{array}
&
\begin{array}{rcl@{\quad}l}
\subsapply{\subs{\termidxv{i}[j]}{u_j}_{j \in J}}{(\termapp{t}{s})}   & \eqdef  & \termapp{\subsapply{\subs{\termidxv{i}[j]}{u_j}_{j \in J}}{t}}{\subsapply{\subs{\termidxv{i}[j]}{u_j}_{j \in J}}{s}} \\
\subsapply{\subs{\termidxv{i}[j]}{u_j}_{j \in J}}{\termabs[n]{p}{t}}  & \eqdef  & \termabs[n]{\subsapply{\subs{\termidxv{i}[j]}{\fincm{u_j}}_{j \in J}}{p}}{\subsapply{\subs{\termidxv{(i+1)}[j]}{\fincv{u_j}}_{j \in J}}{t}}
\end{array}
\end{array} \] It is worth noticing that the base case for variable indices is
undefined if $i' = i$ and $k \notin J$. Such case will render the result of
the substitution undefined as well. In the operational semantics of
$\calcPPCX$, the matching operation presented below will be responsible for
avoiding this undesired situation, as we will see later. The \emph{domain} of a
substitution at level $i$ is given by
$\dom{\subs{\termidxv{i}[j]}{u_j}_{j \in J}} \eqdef
\set{\termidxv{i}[j]}_{j \in J}$. The identity substitution (\ie with empty
domain) is denoted $\subsid$ or $\funcid$.

As in $\calcPPC$, a match ($\mu, \nu, \ldots$) may succeed, fail ($\matchfail$)
or be undetermined ($\matchundet$). For $\calcPPCX$, a successful match will
yield a substitution at level 1, as given by the following \emph{matching
operation}, where the rules are applied in order as in $\calcPPC$: \[
\begin{array}{rcl@{\qquad}l}
\match[n]{\termidxm{1}[j]}{u}                 & \eqdef  & \subs{\termidxv{1}[j]}{u} \\
\match[n]{\termidxm{i+1}[j]}{\termidxm{i}[j]} & \eqdef  & \subsid \\
\match[n]{\termapp{p}{q}}{\termapp{t}{u}}     & \eqdef  & \matchunion{\match[n]{p}{t}}{\match[n]{q}{u}} & \text{if $\termapp{t}{u}, \termapp{p}{q} \in \Matchable{\calcPPCX}$} \\
\match[n]{p}{u}                               & \eqdef  & \matchfail                                    & \text{if $u, p \in \Matchable{\calcPPCX}$} \\
\match[n]{p}{u}                               & \eqdef  & \matchundet                                   & \text{otherwise}
\end{array}
\] where disjoint union of matching is adapted to $\calcPPCX$ from $\calcPPC$
in a straightforward way. The first two rules in the matching operation for
$\calcPPCX$ are worth a comment. As the matching operation should be understood
in the context of a redex, the matchable symbols bound in the pattern are those
with primary index equal to 1. Thus, $\calcPPCX$'s counter-part of the
membership check $x \in \theta$ from $\calcPPC$'s matching operation is a
simple syntactic check on the primary index. Similarly, $x \notin \theta$
corresponds to the primary index being greater than 1, as checked by the second
rule of the definition. However, a primary index $\termidxm{i+1}$ within the
pattern should match primary index $\termidxm{i}$ from the argument, since the
former is affected by an extra binder in a redex. For instance, in the term
$\termapp{(\termabs[1]{\termapp{\termidxm{2}[1]}{\termidxm{1}[1]}}{\termidxv{1}[1]})}{(\termapp{\termidxm{1}[1]}{t'})}$
the matchable $\termidxm{2}[1]$ is free and corresponds to $\termidxm{1}[1]$ in
the argument, while $\termidxm{1}[1]$ from the pattern is bound by the
abstraction. Its counter-part in $\calcPPC$ would be $\alpha$-equivalent to
$\termapp{(\termabs[\lista{x}]{\termapp{\termmatch{y}}{\termmatch{x}}}{\termvar{x}})}{(\termapp{\termmatch{y}}{t})}$.
Hence,
$\match[1]{\termapp{\termidxm{2}[1]}{\termidxm{1}[1]}}{\termapp{\termidxm{1}[1]}{t'}}
= \subs{\termidxv{1}[1]}{t'}$.

As for $\calcPPC$, an additional post-condition is checked over
$\match[n]{p}{u}$ to prevent indices from going out of scope. It requires
$\dom{\match[n]{p}{u}} = \set{\termidx{1}[1], \ldots, \termidx{1}[n]}$, which
essentially implies that all the bound indices are assigned a value by the
resulting substitution. This condition can be guaranteed by requesting
$\set{\termidx{1}[1], \ldots, \termidx{1}[n]} \subseteq \fm{p}$, for each
abstraction $\termabs[n]{p}{t}$ within a well-formed term. To illustrate the
need of such a check, consider the term
$\termapp{(\termabs[2]{\termidxm{1}[1]}{\termidxv{1}[2]})}{u'}$ (\ie the
$\calcPPCX$ counter-part of
$\termapp{(\termabs[\lista{x,y}]{\termmatch{x}}{\termvar{y}})}{u}$, given in
Sec.~\ref{s:preliminaries:ppc}). If the matching
$\match[2]{\termidxm{1}[1]}{u'} = \subs{\termidxv{1}[1]}{u'}$ is considered
correct, then no replacement for the variable index $\termidxv{1}[2]$ in the
body of the abstraction is set, resulting in an ill-behaved operational
semantics.

\medskip
The reduction relation $\reduce[\rPPCX]$ of $\calcPPCX$ is given by the closure
by contexts of the rewriting rule: \[
\termapp{(\termabs[n]{p}{s})}{u} \rrule{\rPPCX} \fdecv{\matchapply{\match[n]{p}{\fincv{u}}}{s}}
\] whenever $\match[n]{p}{\fincv{u}}$ is a decided match. The decrement
function for variable indices is applied to the reduct to compensate for the
loss of a binder over $s$. However, the variable indices of $u$ are not
affected by such binder in the redex. Hence the need of incrementing them prior
to the (eventual) substitution.

Following the reduction example given above for $\calcPPC$, consider these
codifications of $\mathsf{elim} =
\termabs[\lista{x}]{\termmatch{x}}{(\termabs[\lista{y}]{\termapp{\termvar{x}}{\termmatch{y}}}{\termvar{y}})}$
and
$\termabs[\lista{z}]{\termmatch{z}}{\termapp{\termapp{\termmatch{c}}{\termvar{z}}}{\termmatch{n}}}$
respectively:
$\termabs[1]{\termidxm{1}[1]}{(\termabs[1]{\termapp{\termidxv{1}[1]}{\termidxm{1}[1]}}{\termidxv{1}[1]})}$
and
$\termabs[1]{\termidxm{1}[1]}{\termapp{\termapp{\termidxm{1}[1]}{\termidxv{1}[1]}}{\termidxm{2}[1]}}$.
Note how the first occurrence of $\termidxv{1}[1]$ is actually bound by the
outermost abstraction, since abstractions do not bind variable indices in
their pattern. Similarly, the matchable index $\termidxm{1}[1]$ in the body
of
$\termabs[1]{\termidxm{1}[1]}{\termapp{\termapp{\termidxm{1}[1]}{\termidxv{1}[1]}}{\termidxm{2}[1]}}$
turns out to be free as well as $\termidxm{2}[1]$. Then, as expected, we have
the following sequence: \[
\termapp{(\termabs[1]{\termidxm{1}[1]}{(\termabs[1]{\termapp{\termidxv{1}[1]}{\termidxm{1}[1]}}{\termidxv{1}[1]})})}{(\termabs[1]{\termidxm{1}[1]}{\termapp{\termapp{\termidxm{1}[1]}{\termidxv{1}[1]}}{\termidxm{2}[1]}})}
\quad\reduce[\rPPCX]\quad
\termabs[1]{\termapp{(\termabs[1]{\termidxm{1}[1]}{\termapp{\termapp{\termidxm{2}[1]}{\termidxv{1}[1]}}{\termidxm{3}[1]}})}{\termidxm{1}[1]}}{\termidxv{1}[1]}
\quad\reduce[\rPPCX]\quad
\termabs[1]{\termapp{\termapp{\termidxm{2}[1]}{\termidxm{1}[1]}}{\termidxm{3}[1]}}{\termidxv{1}[1]}
\] In the first step,
$\termabs[1]{\termidxm{1}[1]}{\termapp{\termapp{\termidxm{1}[1]}{\termidxv{1}[1]}}{\termidxm{2}[1]}}$
is substituted for $\termidxv{1}[1]$ into the pattern
$\termapp{\termidxv{1}[1]}{\termidxm{1}[1]}$. The fact that the substitution
takes place within the context of a pattern forces the application of
$\fincm{\_}$, thus updating the matchable indices and obtaining
$\termabs[1]{\termidxm{1}[1]}{\termapp{\termapp{\termidxm{2}[1]}{\termidxv{1}[1]}}{\termidxm{3}[1]}}$.
Note that the increment and decrement added by the reduccion rule take no
effect as there are no free variable indices in the term. In the second step,
the resulting application is reduced, giving place to a term whose counter-part
in $\calcPPC$ would be equivalent to
$\termabs[\lista{y}]{\termapp{\termapp{\termmatch{c}}{\termmatch{y}}}{\termmatch{n}}}{\termvar{y}}$
(\cf the reduction example in Sec.~\ref{s:preliminaries:ppc}).

\medskip
In the following sections $\calcPPCX$ is shown to be equivalent to $\calcPPC$
in terms of expressive power and operational semantics. The main advantage of
this new presentation is that it gets rid of $\alpha$-conversion, since there
is no possible collision between free and bound variables/matchables. However,
there is one minor drawback with respect to the use of de Bruijn indices for
the standard $\calcLambda$-calculus. As mentioned above, when working with de
Bruijn indices in the standard $\calcLambda$-calculus, $\alpha$-equivalence
becomes syntactical equality.

Unfortunately, this is not the case when working with bidimensional indices.
For instance, consider the terms
$\termabs[2]{\termapp{\termidxm{1}[1]}{\termidxm{1}[2]}}{\termidxv{1}[1]}$ and
$\termabs[2]{\termapp{\termidxm{1}[2]}{\termidxm{1}[1]}}{\termidxv{1}[2]}$.
Both represent the function that decomposes an application and projects its
first component. But they differ in the way the secondary indices are assigned.
Moreover, one may be tempted to impose an order for the way the secondary
indices are assigned within the pattern to avoid this situation (recall that
the post-condition of the matching operation forces all bound symbols to appear
in the pattern). Given the dynamic nature of patterns in the $\calcPPC$
framework, this enforcement would not solve the problem since patterns may
reduce and such an order is not closed under reduction. For example, consider
$\termabs[2]{\termapp{(\termabs[2]{\termapp{\termidxm{1}[1]}{\termidxm{1}[2]}}{\termapp{\termidxv{1}[2]}{\termidxv{1}[1]}})}{(\termapp{\termidxm{1}[1]}{\termidxm{1}[2]})}}{\termidxv{1}[1]}
\reduce[\rPPCX]
\termabs[2]{\termapp{\termidxm{1}[2]}{\termidxm{1}[1]}}{\termidxv{1}[1]}$.

Fortunately enough, this does not represent a problem from the implementation
point of view, since the ambiguity is local to a binder and does not imply the
need for ``renaming'' variables/matchables while reducing a term, \ie no
possible undesired capture can happen because of it. It is important to note
though, that in the sequel, when refering to equality over terms of
$\calcPPCX$, it is not syntactical equality but equality modulo these
assignments for secondary indices that we are using.

%%% Local Variables:
%%% mode: latex
%%% TeX-master: "main"
%%% End:

%%%%%%%%%%%%%%%%%%%%%%%%%%%%%%%%%%%%%%%%%%%%%%%%%%%%%%%%%%%%%%%%%%%%%%%%%%%%%%%
\section{Translation}
\label{s:translation}
%%%%%%%%%%%%%%%%%%%%%%%%%%%%%%%%%%%%%%%%%%%%%%%%%%%%%%%%%%%%%%%%%%%%%%%%%%%%%%%

This section introduces translations between $\calcPPC$ and $\calcPPCX$ (back
and forth). The goal is to show that these interpretations are suitable to
simulate one calculus into the other. Moreover, the proposed translations turn
out to be the inverse of each other (modulo $\alpha$-conversion) and, as we
will see in Sec.~\ref{s:bisimulation}, they allow to formalise a strong
bisimulation between the two calculi.

\medskip
We start with the translation from $\calcPPC$ to $\calcPPCX$. It takes the
term to be translated together with two lists of lists of symbols that dictate
how the variables and matchables of the terms should be interpreted
respectively. We use lists of lists since the first dimension indicates the
distance to the binder, while the second identifies the symbol among the
multiple bound ones.

Given the lists of lists $X$ and $Y$, we denote by $\listconc{X}{Y}$ their
concatenation. To improve readability, when it is clear from context, we also
write $\listconc{\theta}{X}$ with $\theta$ a list of symbols to denote
$\listconc{\lista{\theta}}{X}$ where $\lista{\_}$ denotes the list constructor.
We use set operations like union and intersection over lists to denote the
union/intersection of its underlying sets.

\begin{definition}
Given a term $t \in \Term{\calcPPC}$ and lists of lists of symbols $V$ and $M$
such that $\fv{t} \subseteq \bigcup_{V' \in V}{V'}$ and $\fm{t} \subseteq
\bigcup_{M' \in M}{M'}$, the \emph{translation of $t$ relative to $V$ and $M$},
written $\toPPCX[V][M]{t}$, is inductively defined as follows: \[
\begin{array}{rcl@{\quad}l}
\toPPCX[V][M]{\termvar{x}}            & \eqdef  & \termidxv{i}[j] & \text{where $i = \min\set{i' \mid x \in V_{i'}}$ and $j = \min\set{j' \mid x = V_{ij'}}$} \\
\toPPCX[V][M]{\termmatch{x}}          & \eqdef  & \termidxm{i}[j] & \text{where $i = \min\set{i' \mid x \in M_{i'}}$ and $j = \min\set{j' \mid x = M_{ij'}}$} \\
\toPPCX[V][M]{\termapp{t}{u}}         & \eqdef  & \termapp{\toPPCX[V][M]{t}}{\toPPCX[V][M]{u}} \\
\toPPCX[V][M]{\termabs[\theta]{p}{t}} & \eqdef  & \termabs[\fsize{\theta}]{\toPPCX[V][\listconc{\theta}{M}]{p}}{\toPPCX[\listconc{\theta}{V}][M]{t}}
\end{array}
\] Let $x_1, x_2, \ldots$ be an enumeration of $\TermVariable$. Then, the
\emph{translation} of $t$ to $\calcPPCX$, written simply $\toPPCX{t}$, is
defined as $\toPPCX[X][X]{t}$ where $X = \lista{\lista{x_1}, \ldots,
\lista{x_n}}$ such that $\fv{t} \cup \fm{t} \subseteq \set{x_1, \ldots, x_n}$.
\label{d:translation:to-ppcx}
\end{definition}

For example, consider the term $s_0 =
\termapp{(\termabs[\lista{x}]{\termapp{\termmatch{y}}{\termmatch{x}}}{\termvar{x}})}{(\termapp{\termmatch{y}}{s'_0})}$
with $\fv{s'_0} = \set{y}$ and $\fm{s'_0} = \set{y}$. Then,
$\toPPCX[\lista{\lista{y}}][\lista{\lista{y}}]{s_0} =
\termapp{(\termabs[1]{\termapp{\toPPCX[\lista{\lista{y}}][\lista{\lista{x},\lista{y}}]{\termmatch{y}}}{\toPPCX[\lista{\lista{y}}][\lista{\lista{x},\lista{y}}]{\termmatch{x}}}}{\toPPCX[\lista{\lista{x},\lista{y}}][\lista{\lista{y}}]{\termvar{x}}})}{(\termapp{\toPPCX[\lista{\lista{y}}][\lista{\lista{y}}]{\termmatch{y}}}{\toPPCX[\lista{\lista{y}}][\lista{\lista{y}}]{s'_0}})}
=
\termapp{(\termabs[1]{\termapp{\termidxm{2}[1]}{\termidxm{1}[1]}}{\termidxv{1}[1]})}{(\termapp{\termidxm{1}[1]}{\toPPCX[\lista{\lista{y}}][\lista{\lista{y}}]{s'_0}})}$.
Note that the inicialisation of $V$ and $M$ with singleton elements implies
that each free variable/matchable in the term will be assigned a distinct
primary index (when interpreted at the same depth), following de Bruijn's
original ideas: let $s_1 =
\termapp{(\termabs[\lista{x}]{\termapp{\termmatch{y}}{\termmatch{x}}}{\termvar{x}})}{(\termapp{\termmatch{y}}{\termvar{z}})}$,
then $\toPPCX[\lista{\lista{y},\lista{z}}][\lista{\lista{y},\lista{z}}]{s_1} =
\termapp{(\termabs[1]{\termapp{\termidxm{2}[1]}{\termidxm{1}[1]}}{\termidxv{1}[1]})}{(\termapp{\termidxm{1}[1]}{\termidxv{2}[1]})}$.

\medskip
Our main goal is to prove that $\calcPPCX$ simulates $\calcPPC$ via this
embedding. For this purpose we need to state first some auxiliary lemmas that
prove how the translation behaves with respect to the substitution and the
matching operation. We start with a technical result concerning the increment
functions for variable and matchable indices. Notation $\fincv[k]{t}[n]$ stands
for $n$ consecutive applications of $\fincv[k]{\_}$ over $t$ (similarly for
$\fincm[k]{t}[n]$).

\begin{lemma}
Let $t \in \Term{\calcPPC}$, $k \geq 0$, $i \geq 1$ and $n \geq k + i$ such
that $X_l \cap (\fv{t} \cup \fm{t}) = \emptyset$ for all $l \in [k+1,k+i-1]$.
Then,
\begin{enumerate}
  \item\label{l:translation:to-ppcx-upd:var}
  $\toPPCX[\listconc{\listconc{X_1}{\ldots}}{X_n}][M]{t} =
  \fincv[k]{\toPPCX[\listconc{\listconc{\listconc{X_1}{\ldots}}{X_k}}{\listconc{\listconc{X_{k+i}}{\ldots}}{X_n}}][M]{t}}[i-1]$.

  \item\label{l:translation:to-ppcx-upd:match}
  $\toPPCX[V][\listconc{\listconc{X_1}{\ldots}}{X_n}]{t} =
  \fincm[k]{\toPPCX[V][\listconc{\listconc{\listconc{X_1}{\ldots}}{X_k}}{\listconc{\listconc{X_{k+i}}{\ldots}}{X_n}}]{t}}[i-1]$.
\end{enumerate}
\label{l:translation:to-ppcx-upd}
\end{lemma}

\begin{proof}
\begin{enumerate}
  \item By induction on $t$.
  \begin{itemize}
    \item $t = \termvar{x}$. Let
    $\toPPCX[\listconc{\listconc{X_1}{\ldots}}{X_n}][M]{\termvar{x}} =
    \termidxv{i'}[j]$, \ie $x = X_{i'j}$. By hypothesis, $i' \notin [k+1,k+i-1]$.
    Then, there are two possible cases:
    \begin{enumerate}
      \item $i' \leq k$. Then,
      $\toPPCX[\listconc{\listconc{X_1}{\ldots}}{X_n}][M]{\termvar{x}}
      = \termidxv{i'}[j] = \fincv[k]{\termidxv{i'}[j]}[i-1] =
      \fincv[k]{\toPPCX[\listconc{\listconc{\listconc{X_1}{\ldots}}{X_k}}{\listconc{\listconc{X_{k+i}}{\ldots}}{X_n}}][M]{\termvar{x}}}[i-1]$.
      
      \item $i' \geq k+i$. Then, $i' - (i - 1) > k$ and hence we have
      $\termidxv{i'}[j] = \fincv[k]{\termidxv{(i'-(i-1))}[j]}[i-1]$. Thus, we
      conclude $\toPPCX[\listconc{\listconc{X_1}{\ldots}}{X_n}][M]{\termvar{x}}
      = \termidxv{i'}[j] = \fincv[k]{\termidxv{(i'-(i-1))}[j]}[i-1] =
      \fincv[k]{\toPPCX[\listconc{\listconc{\listconc{X_1}{\ldots}}{X_k}}{\listconc{\listconc{X_{k+i}}{\ldots}}{X_n}}][M]{\termvar{x}}}[i-1]$.
    \end{enumerate}
    
    \item $t = \termmatch{x}$. This is immediate since
    $\toPPCX[\listconc{\listconc{X_1}{\ldots}}{X_n}][M]{\termvar{x}} =
    \termidxm{i'}[j] =
    \toPPCX[\listconc{\listconc{\listconc{X_1}{\ldots}}{X_k}}{\listconc{\listconc{X_{k+i}}{\ldots}}{X_n}}][M]{\termvar{x}}$
    for some $i',j \in \Natural_{\geq 1}$ such that $x = M_{i'j}$, and
    $\fincv[k]{\_}[i-1]$ leaves matchable indices untouched.
    
    \item $t = \termapp{s}{u}$. By \ih we have
    $\toPPCX[\listconc{\listconc{X_1}{\ldots}}{X_n}][M]{s} =
    \fincv[k]{\toPPCX[\listconc{\listconc{\listconc{X_1}{\ldots}}{X_k}}{\listconc{\listconc{X_{k+i}}{\ldots}}{X_n}}][M]{s}}[i-1]$
    and $\toPPCX[\listconc{\listconc{X_1}{\ldots}}{X_n}][M]{u} =
    \fincv[k]{\toPPCX[\listconc{\listconc{\listconc{X_1}{\ldots}}{X_k}}{\listconc{\listconc{X_{k+i}}{\ldots}}{X_n}}][M]{u}}[i-1]$.
    Thus, we conclude $\toPPCX[\listconc{\listconc{X_1}{\ldots}}{X_n}][M]{t} =
    \termapp{\toPPCX[\listconc{\listconc{X_1}{\ldots}}{X_n}][M]{s}}{\toPPCX[\listconc{\listconc{X_1}{\ldots}}{X_n}][M]{u}}
    = 
    \termapp{\fincv[k]{\toPPCX[\listconc{\listconc{\listconc{X_1}{\ldots}}{X_k}}{\listconc{\listconc{X_{k+i}}{\ldots}}{X_n}}][M]{s}}[i-1]}{\fincv[k]{\toPPCX[\listconc{\listconc{\listconc{X_1}{\ldots}}{X_k}}{\listconc{\listconc{X_{k+i}}{\ldots}}{X_n}}][M]{u}}[i-1]}
    =
    \fincv[k]{\toPPCX[\listconc{\listconc{\listconc{X_1}{\ldots}}{X_k}}{\listconc{\listconc{X_{k+i}}{\ldots}}{X_n}}][M]{t}}[i-1]$.
    
    \item $t = \termabs[\theta]{p}{s}$. By \ih
    $\toPPCX[\listconc{\listconc{X_1}{\ldots}}{X_n}][\listconc{\theta}{M}]{p} =
    \fincv[k]{\toPPCX[\listconc{\listconc{\listconc{X_1}{\ldots}}{X_k}}{\listconc{\listconc{X_{k+i}}{\ldots}}{X_n}}][\listconc{\theta}{M}]{p}}[i-1]$
    and $\toPPCX[\listconc{\theta}{\listconc{\listconc{X_1}{\ldots}}{X_n}}][M]{s} =
    \fincv[k+1]{\toPPCX[\listconc{\theta}{\listconc{\listconc{\listconc{X_1}{\ldots}}{X_k}}{\listconc{\listconc{X_{k+i}}{\ldots}}{X_n}}}][M]{s}}[i-1]$.
    Note that $\theta$ is pushed accordingly in the list of variable or
    matchable symbols for $s$ and $p$ respectively, following the definition of
    $\toPPCX{\_}$ for abstractions. In the case for $s$, this implies
    concluding with $\fincv[k+1]{\_}$ instead of $\fincv[k]{\_}$. Finally,
    conclude by Def.~\ref{d:translation:to-ppcx},
    $\toPPCX[\listconc{\listconc{X_1}{\ldots}}{X_n}][M]{t} =
    \termabs[\fsize{\theta}]{\toPPCX[\listconc{\listconc{X_1}{\ldots}}{X_n}][\listconc{\theta}{M}]{p}}{\toPPCX[\listconc{\theta}{\listconc{\listconc{X_1}{\ldots}}{X_n}}][M]{s}}
    = 
    \termabs[\fsize{\theta}]{\fincv[k]{\toPPCX[\listconc{\listconc{\listconc{X_1}{\ldots}}{X_k}}{\listconc{\listconc{X_{k+i}}{\ldots}}{X_n}}][\listconc{\theta}{M}]{p}}[i-1]}{\fincv[k+1]{\toPPCX[\listconc{\theta}{\listconc{\listconc{\listconc{X_1}{\ldots}}{X_k}}{\listconc{\listconc{X_{k+i}}{\ldots}}{X_n}}}][M]{s}}[i-1]}
    =
    \fincv[k]{\toPPCX[\listconc{\listconc{\listconc{X_1}{\ldots}}{X_k}}{\listconc{\listconc{X_{k+i}}{\ldots}}{X_n}}][M]{t}}[i-1]$.
  \end{itemize}
  
  \item By induction on $t$. This item is similar to the previous one.
\end{enumerate}
\end{proof}

%%% Local Variables: 
%%% mode: latex
%%% TeX-master: "main"
%%% End: 

The translation of a substitution $\sigma$ requires an enumeration $\theta$
such that $\dom{\sigma} \subseteq \theta$ to be provided. It is then defined as
$\toPPCX[V][M]{\sigma}[\theta] \eqdef
\subs{\termidxv{1}[j]}{\toPPCX[V][M]{\subsapply{\sigma}{x_j}}}_{x_j \in \dom{\sigma}}$.
Note how substitutions from $\calcPPC$ are mapped into substitutions at level 1 in
the $\calcPPCX$ framework. This suffices since substitutions are only meant to
be created in the context of a redex. When acting at arbitrary depth on a term,
substitutions are shown to behave properly.

\begin{lemma}
Let $s \in \Term{\calcPPC}$, $\sigma$ be a substitution,
$\theta$ be an enumeration such that $\dom{\sigma}
\subseteq \theta$ and $Y$ be a list of $i-1$ lists of symbols such that
$(\bigcup_{Y' \in Y}{Y'}) \cap \theta = \emptyset$ and
$(\bigcup_{Y' \in Y}{Y'}) \cap
(\bigcup_{x_j \in \theta}{\fv{\subsapply{\sigma}{\termvar{x_j}}}}) =
\emptyset$. Then,
$\toPPCX[\listconc{Y}{X}][M]{\subsapply{\sigma}{s}} =
\fdecv[i-1]{\subsapply{\subs{\termidxv{i}[j]}{\fincv[i-1]{\toPPCX[\listconc{Y}{X}][M]{\subsapply{\sigma}{\termvar{x_j}}}}}_{x_j \in \dom{\sigma}}}{\toPPCX[\listconc{\listconc{Y}{\theta}}{X}][M]{s}}}$.
\label{l:translation:to-ppcx-subs}
\end{lemma}

\begin{proof}
By induction on term $s$.
\begin{itemize}
  \item $s = \termvar{y}$. There are three possible cases.
  \begin{enumerate}
    \item $y \in \theta$ (\ie $y = \theta_k$ for some $k$). By hypothesis,
    $\theta_k \notin \bigcup_{Y' \in Y}{Y'}$ and
    $\fv{\subsapply{\sigma}{\termvar{\theta_k}}} \cap \bigcup_{Y' \in Y}{Y'} =
    \emptyset$. Moreover, let $V = \listconc{\listconc{Y}{\theta}}{X}$, then $i
    = \min\set{i' \mid \theta_k \in V_{i'}}$ and $k =
    \min\set{j' \mid \theta_k = V_{ij'}}$. Thus,
    $\toPPCX[\listconc{\listconc{Y}{\theta}}{X}][M]{\termvar{\theta_k}} =
    \termidxv{i}[k]$. We conclude since
    $\fdecv[i-1]{\subsapply{\subs{\termidxv{i}[j]}{\fincv[i-1]{\toPPCX[\listconc{Y}{X}][M]{\subsapply{\sigma}{\termvar{x_j}}}}}_{x_j \in \dom{\sigma}}}{\toPPCX[\listconc{\listconc{Y}{\theta}}{X}][M]{\termvar{\theta_k}}}}
    =
    \fdecv[i-1]{\fincv[i-1]{\toPPCX[\listconc{Y}{X}][M]{\subsapply{\sigma}{\termvar{\theta_k}}}}}
    = \toPPCX[\listconc{Y}{X}][M]{\subsapply{\sigma}{\termvar{\theta_k}}}$ by
    Lem.~\ref{l:appendix:finc} (\ref{l:appendix:finc:fdec-l}).
    
    \item $y \in \bigcup_{Y' \in Y}{Y'}$. By hypothesis, $y \notin \theta$.
    Then, $\subsapply{\sigma}{\termvar{y}} = \termvar{y}$. Let
    $\toPPCX[\listconc{\listconc{Y}{\theta}}{X}][M]{\termvar{y}} =
    \termidxv{i'}[k]$.
    Moreover, $i' < i$. Thus, we conclude
    $\fdecv[i-1]{\subsapply{\subs{\termidxv{i}[j]}{\fincv[i-1]{\toPPCX[\listconc{Y}{X}][M]{\subsapply{\sigma}{\termvar{x_j}}}}}_{x_j \in \dom{\sigma}}}{\toPPCX[\listconc{\listconc{Y}{\theta}}{X}][M]{\termvar{y}}}}
    = \fdecv[i-1]{\termidxv{i'}[k]} = \termidxv{i'}[k] =
    \toPPCX[\listconc{Y}{X}][M]{\subsapply{\sigma}{\termvar{y}}}$ since $i' < i$.
    
    \item Otherwise, \ie $y \in \bigcup_{X' \in X}{X'}$, $y \notin
    \bigcup_{Y' \in Y}{Y'}$ and $y \notin \theta$. Then,
    $\subsapply{\sigma}{\termvar{y}} = \termvar{y}$. Let
    $\toPPCX[\listconc{\listconc{Y}{\theta}}{X}][M]{\termvar{y}} =
    \termidxv{i'}[k]$. Moreover, $i' > i$ and
    $\toPPCX[\listconc{Y}{X}][M]{\termvar{y}} = \termidxv{(i'-1)}[k]$. Thus, we
    conclude
    $\fdecv[i-1]{\subsapply{\subs{\termidxv{i}[j]}{\fincv[i-1]{\toPPCX[\listconc{Y}{X}][M]{\subsapply{\sigma}{\termvar{x_j}}}}}_{x_j \in \dom{\sigma}}}{\toPPCX[\listconc{\listconc{Y}{\theta}}{X}][M]{\termvar{y}}}}
    = \fdecv[i-1]{\termidxv{i'}[k]} = \termidxv{(i'-1)}[k] = 
    \toPPCX[\listconc{Y}{X}][M]{\subsapply{\sigma}{\termvar{y}}}$ since $i' >
    i$.
  \end{enumerate}
  
  \item $s = \termmatch{y}$. Then,
  $\toPPCX[\listconc{\listconc{Y}{\theta}}{X}][M]{\termmatch{y}} =
  \termidxm{i'}[k]$ for some $i',k \in \Natural_{\geq 1}$. Thus,
  $\fdecv[i-1]{\subsapply{\subs{\termidxv{i}[j]}{\fincv[i-1]{\toPPCX[\listconc{Y}{X}][M]{\subsapply{\sigma}{\termvar{x_j}}}}}_{x_j \in \dom{\sigma}}}{\toPPCX[\listconc{\listconc{Y}{\theta}}{X}][M]{\termmatch{y}}}}
  = \fdecv[i-1]{\termidxm{i'}[k]} = \termidxm{i'}[k] =
  \toPPCX[\listconc{Y}{X}][M]{\subsapply{\sigma}{\termmatch{y}}}$ and we
  conclude.
  
  \item $s = \termapp{t}{u}$. This case is immediate from the \ih since every
  definition involved distributes over applications.
  
  \item $s = \termabs[\theta']{p}{t}$. W.l.o.g. we assume
  $\pavoids{\theta'}{\sigma}$ and also $\theta'$ fresh for $\theta$, $M$, $X$
  and $Y$. By \ih we have
  $\toPPCX[\listconc{Y}{X}][\listconc{\theta'}{M}]{\subsapply{\sigma}{p}} =
  \fdecv[i-1]{\subsapply{\subs{\termidxv{i}[j]}{\fincv[i-1]{\toPPCX[\listconc{Y}{X}][\listconc{\theta'}{M}]{\subsapply{\sigma}{\termvar{x_j}}}}}_{x_j \in \dom{\sigma}}}{\toPPCX[\listconc{\listconc{Y}{\theta}}{X}][\listconc{\theta'}{M}]{p}}}$
  and
  $\toPPCX[\listconc{\listconc{\theta'}{Y}}{X}][M]{\subsapply{\sigma}{t}} =
  \fdecv[i]{\subsapply{\subs{\termidxv{(i+1)}[j]}{\fincv[i]{\toPPCX[\listconc{\listconc{\theta'}{Y}}{X}][M]{\subsapply{\sigma}{\termvar{x_j}}}}}_{x_j \in \dom{\sigma}}}{\toPPCX[\listconc{\listconc{\listconc{\theta'}{Y}}{\theta}}{X}][M]{t}}}$.
  Moreover, by Lem.~\ref{l:translation:to-ppcx-upd}
  (\ref{l:translation:to-ppcx-upd:match}), we get
  $\toPPCX[\listconc{Y}{X}][\listconc{\theta'}{M}]{\subsapply{\sigma}{p}} =
  \fdecv[i-1]{\subsapply{\subs{\termidxv{i}[j]}{\fincv[i-1]{\fincm{\toPPCX[\listconc{Y}{X}][M]{\subsapply{\sigma}{\termvar{x_j}}}}}}_{x_j \in \dom{\sigma}}}{\toPPCX[\listconc{\listconc{Y}{\theta}}{X}][\listconc{\theta'}{M}]{p}}}$.
  Similarly, by Lem.~\ref{l:translation:to-ppcx-upd}
  (\ref{l:translation:to-ppcx-upd:var}),
  $\toPPCX[\listconc{\listconc{\theta'}{Y}}{X}][M]{\subsapply{\sigma}{t}} =
  \fdecv[i]{\subsapply{\subs{\termidxv{(i+1)}[j]}{\fincv[i]{\fincv{\toPPCX[\listconc{Y}{X}][M]{\subsapply{\sigma}{\termvar{x_j}}}}}}_{x_j \in \dom{\sigma}}}{\toPPCX[\listconc{\listconc{\listconc{\theta'}{Y}}{\theta}}{X}][M]{t}}}$.
  Furthermore, by Lem.~\ref{l:appendix:finc} (\ref{l:appendix:finc:match}),
  $\fincv[i-1]{\fincm{\toPPCX[\listconc{Y}{X}][M]{\subsapply{\sigma}{\termvar{x_j}}}}}
  =
  \fincm{\fincv[i-1]{\toPPCX[\listconc{Y}{X}][M]{\subsapply{\sigma}{\termvar{x_j}}}}}$
  and, by Lem.~\ref{l:appendix:finc} (\ref{l:appendix:finc:var}),
  $\fincv[i]{\fincv{\toPPCX[\listconc{Y}{X}][M]{\subsapply{\sigma}{\termvar{x_j}}}}}
  =
  \fincv{\fincv[i-1]{\toPPCX[\listconc{Y}{X}][M]{\subsapply{\sigma}{\termvar{x_j}}}}}$.
  Finally, we conclude by definition of $\fdecv{\_}$ and
  Def.~\ref{d:translation:to-ppcx}: \[%\kern-3em
\begin{array}{ll}
  & \toPPCX[\listconc{Y}{X}][M]{\subsapply{\sigma}{(\termabs[\theta']{p}{t})}} \\
= & \termabs[\fsize{\theta'}]{\toPPCX[\listconc{Y}{X}][\listconc{\theta'}{M}]{\subsapply{\sigma}{p}}}{\toPPCX[\listconc{\listconc{\theta'}{Y}}{X}][M]{\subsapply{\sigma}{t}}} \\
= & \termabs[\fsize{\theta'}]{\fdecv[i-1]{\subsapply{\subs{\termidxv{i}[j]}{\fincm{\fincv[i-1]{\toPPCX[\listconc{Y}{X}][M]{\subsapply{\sigma}{\termvar{x_j}}}}}}_{x_j \in \dom{\sigma}}}{\toPPCX[\listconc{\listconc{Y}{\theta}}{X}][\listconc{\theta'}{M}]{p}}}}{\fdecv[i]{\subsapply{\subs{\termidxv{(i+1)}[j]}{\fincv{\fincv[i-1]{\toPPCX[\listconc{Y}{X}][M]{\subsapply{\sigma}{\termvar{x_j}}}}}}_{x_j \in \dom{\sigma}}}{\toPPCX[\listconc{\listconc{\listconc{\theta'}{Y}}{\theta}}{X}][M]{t}}}} \\
= & \fdecv[i-1]{\termabs[\fsize{\theta'}]{\subsapply{\subs{\termidxv{i}[j]}{\fincm{\fincv[i-1]{\toPPCX[\listconc{Y}{X}][M]{\subsapply{\sigma}{\termvar{x_j}}}}}}_{x_j \in \dom{\sigma}}}{\toPPCX[\listconc{\listconc{Y}{\theta}}{X}][\listconc{\theta'}{M}]{p}}}{\subsapply{\subs{\termidxv{(i+1)}[j]}{\fincv{\fincv[i-1]{\toPPCX[\listconc{Y}{X}][M]{\subsapply{\sigma}{\termvar{x_j}}}}}}_{x_j \in \dom{\sigma}}}{\toPPCX[\listconc{\listconc{\listconc{\theta'}{Y}}{\theta}}{X}][M]{t}}}} \\
= & \fdecv[i-1]{\subsapply{\subs{\termidxv{i}[j]}{\fincv[i-1]{\toPPCX[\listconc{Y}{X}][M]{\subsapply{\sigma}{\termvar{x_j}}}}}_{x_j \in \dom{\sigma}}}{\termabs[\fsize{\theta'}]{\toPPCX[\listconc{\listconc{Y}{\theta}}{X}][\listconc{\theta'}{M}]{p}}{\toPPCX[\listconc{\listconc{\listconc{\theta'}{Y}}{\theta}}{X}][M]{t}}}} \\
= & \fdecv[i-1]{\subsapply{\subs{\termidxv{i}[j]}{\fincv[i-1]{\toPPCX[\listconc{Y}{X}][M]{\subsapply{\sigma}{\termvar{x_j}}}}}_{x_j \in \dom{\sigma}}}{\toPPCX[\listconc{\listconc{Y}{\theta}}{X}][M]{\termabs[\theta']{p}{t}}}}
\end{array} \]
\end{itemize}
\end{proof}

%%% Local Variables:
%%% mode: latex
%%% TeX-master: "main"
%%% End:

In the case of a match, its translation is given by
$\toPPCX[V][M]{\match[\theta]{p}{u}} \eqdef
\match[\fsize{\theta}]{\toPPCX[V][\listconc{\theta}{M}]{p}}{\toPPCX[V][M]{u}}$.
Note how $\theta$ is pushed into the matchable symbol list of the pattern, in
accordance with the translation of an abstraction. This is crucial for the
following result of preservation of the matching output.

First, note that matchable forms are preserved by the translation as well.

\begin{lemma}
Let $t \in \Term{\calcPPC}$. Then,
\begin{enumerate}
  \item \label{l:appendix:to-ppcx-mf:data} $t \in \TermData{\calcPPC}$ iff
  $\toPPCX[V][M]{t} \in \TermData{\calcPPCX}$.
  \item \label{l:appendix:to-ppcx-mf:match} $t \in \Matchable{\calcPPC}$ iff
  $\toPPCX[V][M]{t} \in \Matchable{\calcPPCX}$.
\end{enumerate}
\label{l:appendix:to-ppcx-mf}
\end{lemma}

\begin{proof}
Both items follow by straightforward induction on $t$, using
(\ref{l:appendix:to-ppcx-mf:data}) to prove (\ref{l:appendix:to-ppcx-mf:match}).
\end{proof}

Then, the result of preservation of the matching output states:

\begin{lemma}
Let $p, u \in \Term{\calcPPC}$.
\begin{enumerate}
  \item\label{l:translation:to-ppcx-match:subs} If $\match[\theta]{p}{u} =
  \sigma$, then $\toPPCX[V][M]{\match[\theta]{p}{u}} =
  \toPPCX[V][M]{\sigma}[\theta]$.

  \item\label{l:translation:to-ppcx-match:fail} If $\match[\theta]{p}{u} =
  \matchfail$, then $\toPPCX[V][M]{\match[\theta]{p}{u}} = \matchfail$.

  \item\label{l:translation:to-ppcx-match:wait} If $\match[\theta]{p}{u} =
  \matchundet$, then $\toPPCX[V][M]{\match[\theta]{p}{u}} = \matchundet$.
\end{enumerate}
\label{l:translation:to-ppcx-match}
\end{lemma}

\begin{proof}
By induction on $p$ considering the result of $\match[\theta]{p}{u}$ before the
final consistency check.
\begin{itemize}
  \item $p = \termvar{x}$. Then, $\toPPCX[V][\listconc{\theta}{M}]{\termvar{x}}
  = \termidxv{i}[j]$ for some $i,j \in \Natural_{\geq 1}$. Moreover,
  $\match[\theta]{p}{u} = \matchundet$. We conclude
  (\ref{l:translation:to-ppcx-match:wait}) since
  $\toPPCX[V][M]{\match[\theta]{\termvar{x}}{u}} = 
  \match[\fsize{\theta}]{\toPPCX[V][\listconc{\theta}{M}]{\termvar{x}}}{\toPPCX[V][M]{u}}
  = \matchundet$ too.
  
  \item $p = \termmatch{x}$. There are two possible cases:
  \begin{enumerate}
    \item $x \in \theta$. Then, 
    $\toPPCX[V][\listconc{\theta}{M}]{\termmatch{x}} = \termidxm{1}[j]$ for some
    $j$ such that $j$ is the index of $x$ in $\theta$. Moreover,
    $\match[\theta]{p}{u} = \subs{\termvar{x}}{u}$. We conclude
    (\ref{l:translation:to-ppcx-match:subs}) since
    $\toPPCX[V][M]{\subs{\termvar{x}}{u}}[\theta] =
    \subs{\termidxv{1}[j]}{\toPPCX[V][M]{u}} =
    \match[\fsize{\theta}]{\toPPCX[V][\listconc{\theta}{M}]{\termmatch{x}}}{\toPPCX[V][M]{u}}
    = \toPPCX[V][M]{\match[\theta]{p}{u}}$.
    
    \item $x \notin \theta$. Then,
    $\toPPCX[V][\listconc{\theta}{M}]{\termmatch{x}} = \termidxm{(i+1)}[j]$ for
    some $i,j \in \Natural_{\geq 1}$. There are two further cases to analyse:
    \begin{enumerate}
      \item $u = \termmatch{x}$. Then, $\toPPCX[V][M]{\termmatch{x}} =
      \termidxm{i}[j]$. Moreover, $\match[\theta]{p}{u} = \subsid$. We conclude
      (\ref{l:translation:to-ppcx-match:subs}) since
      $\toPPCX[V][M]{\subsid}[\theta] = \subsid =
      \match[\fsize{\theta}]{\toPPCX[V][\listconc{\theta}{M}]{\termmatch{x}}}{\toPPCX[V][M]{u}}
      = \toPPCX[V][M]{\match[\theta]{p}{u}}$.
      
      \item $u \neq \termmatch{x}$. By Lem.~\ref{l:appendix:to-ppcx-mf}
      (\ref{l:appendix:to-ppcx-mf:match}), $u \in \Matchable{\calcPPC}$ iff
      $\toPPCX[V][M]{u} \in \Matchable{\calcPPCX}$. Then, we have
      $\match[\theta]{\termmatch{x}}{u} = \matchfail$ implies
      $\toPPCX[V][M]{\match[\theta]{\termmatch{x}}{u}} =
      \match[\fsize{\theta}]{\toPPCX[V][\listconc{\theta}{M}]{\termmatch{x}}}{\toPPCX[V][M]{u}}
      = \matchfail$ too, and we conclude
      (\ref{l:translation:to-ppcx-match:fail}). Moreover,
      $\match[\theta]{\termmatch{x}}{u} = \matchundet$ implies
      $\toPPCX[V][M]{\match[\theta]{\termmatch{x}}{u}} = \matchundet$ as well,
      allowing to conclude (\ref{l:translation:to-ppcx-match:wait}).
    \end{enumerate}
  \end{enumerate}
  
  \item $p = \termapp{q}{r}$. Then, $\toPPCX[V][\listconc{\theta}{M}]{p} =
  \termapp{\toPPCX[V][\listconc{\theta}{M}]{q}}{\toPPCX[V][\listconc{\theta}{M}]{r}}$.
  There are two possible cases:
  \begin{enumerate}
    \item $p \in \Matchable{\calcPPC}$. By Lem.~\ref{l:appendix:to-ppcx-mf}
    (\ref{l:appendix:to-ppcx-mf:match}), $\toPPCX[V][\listconc{\theta}{M}]{p} \in
    \Matchable{\calcPPCX}$. There are two further cases to analyse:
    \begin{enumerate}
      \item $u = \termapp{s}{t}$. By Lem.~\ref{l:appendix:to-ppcx-mf}
      (\ref{l:appendix:to-ppcx-mf:match}), $u \in \Matchable{\calcPPC}$ iff
      $\toPPCX[V][M]{u} \in \Matchable{\calcPPCX}$. Thus, if
      $\match[\theta]{p}{u} = \matchundet$ (\ie $u \notin
      \Matchable{\calcPPC}$), then $\toPPCX[V][M]{\match[\theta]{p}{u}} =
      \matchundet$ as well, allowing to conclude
      (\ref{l:translation:to-ppcx-match:wait}). Otherwise, $\match[\theta]{p}{u}
      = \matchunion{\match[\theta]{q}{s}}{\match[\theta]{r}{t}}$. If either is
      $\matchfail$, the result is immediate from the \ih
      (\ref{l:translation:to-ppcx-match:fail}). Similarly for the case where
      either of the two is $\matchundet$, using \ih
      (\ref{l:translation:to-ppcx-match:wait}). Assume $\match[\theta]{q}{s} =
      \sigma_1$ and $\match[\theta]{r}{t} = \sigma_2$. By \ih we have
      $\toPPCX[V][M]{\match[\theta]{q}{s}} = \toPPCX[V][M]{\sigma_1}[\theta]$
      and $\toPPCX[V][M]{\match[\theta]{r}{t}} =
      \toPPCX[V][M]{\sigma_2}[\theta]$. Note that $\dom{\sigma_1} \cap
      \dom{\sigma_2} = \emptyset$ iff $\dom{\toPPCX[V][M]{\sigma_1}[\theta]}
      \cap \dom{\toPPCX[V][M]{\sigma_2}[\theta]} = \emptyset$ since both
      translations are given the same enumeration $\theta$. Then it is safe to
      conclude $\toPPCX[V][M]{\match[\theta]{p}{u}} =
      \matchunion{\toPPCX[V][M]{\match[\theta]{q}{s}}}{\toPPCX[V][M]{\match[\theta]{r}{t}}}$.
      
      \item $u \neq \termapp{s}{t}$. By Lem.~\ref{l:appendix:to-ppcx-mf}
      (\ref{l:appendix:to-ppcx-mf:match}), $u \in \Matchable{\calcPPC}$ iff
      $\toPPCX[V][M]{u} \in \Matchable{\calcPPCX}$. Then, we have
      $\match[\theta]{p}{u} = \matchfail$ implies
      $\toPPCX[V][M]{\match[\theta]{p}{u}} =
      \match[\fsize{\theta}]{\toPPCX[V][\listconc{\theta}{M}]{p}}{\toPPCX[V][M]{u}}
      = \matchfail$ too, and we conclude
      (\ref{l:translation:to-ppcx-match:fail}). Moreover,
      $\match[\theta]{p}{u} = \matchundet$ implies
      $\toPPCX[V][M]{\match[\theta]{p}{u}} = \matchundet$ as well,
      allowing to conclude (\ref{l:translation:to-ppcx-match:wait}).
    \end{enumerate}
    
    \item $p \notin \Matchable{\calcPPC}$. By Lem.~\ref{l:appendix:to-ppcx-mf}
    (\ref{l:appendix:to-ppcx-mf:match}), $\toPPCX[V][\listconc{\theta}{M}]{p}
    \notin \Matchable{\calcPPCX}$. Then, $\match[\theta]{p}{u} = \matchundet$
    and $\toPPCX[V][M]{\match[\theta]{p}{u}} =
    \match[\fsize{\theta}]{\toPPCX[V][\listconc{\theta}{M}]{p}}{\toPPCX[V][M]{u}}
    = \matchundet$ too. Thus, we conclude
    (\ref{l:translation:to-ppcx-match:wait}).
  \end{enumerate}
  
  \item $p = \termabs[\theta']{q}{s}$. Then, $\toPPCX[V][\listconc{\theta}{M}]{p}
  =
  \termabs[\fsize{\theta'}]{\toPPCX[V][\listconc{\theta'}{M}]{q}}{\toPPCX[\listconc{\theta'}{V}][M]{s}}$.
  Then, $p \in \Matchable{\calcPPC}$ and, by Lem.~\ref{l:appendix:to-ppcx-mf}
  (\ref{l:appendix:to-ppcx-mf:match}), $\toPPCX[V][\listconc{\theta}{M}]{p}
  \in \Matchable{\calcPPCX}$ too. Also by Lem.~\ref{l:appendix:to-ppcx-mf}
  (\ref{l:appendix:to-ppcx-mf:match}), $u \in \Matchable{\calcPPC}$ iff
  $\toPPCX[V][M]{u} \in \Matchable{\calcPPCX}$. Then, we have
  $\match[\theta]{p}{u} = \matchfail$ implies
  $\toPPCX[V][M]{\match[\theta]{p}{u}} =
  \match[\fsize{\theta}]{\toPPCX[V][\listconc{\theta}{M}]{p}}{\toPPCX[V][M]{u}}
  = \matchfail$ as well, and we conclude
  (\ref{l:translation:to-ppcx-match:fail}). Moreover,
  $\match[\theta]{p}{u} = \matchundet$ implies
  $\toPPCX[V][M]{\match[\theta]{p}{u}} = \matchundet$ as well, allowing to
  conclude (\ref{l:translation:to-ppcx-match:wait}).
\end{itemize}
Finally, once the matching operation returns, we need to verify for
(\ref{l:translation:to-ppcx-match:subs}) that $\dom{\sigma} = \theta$ implies
$\dom{\toPPCX[V][M]{\sigma}[\theta]} =
\set{\termidx{1}[1], \ldots, \termidx{1}[\fsize{\theta}]}$. This is precisely
the case since, by definition, $\dom{\toPPCX[V][M]{\sigma}[\theta]} =
\set{\termidx{1}[j]}_{x_j \in \dom{\sigma}}$.
\end{proof}

%%% Local Variables:
%%% mode: latex
%%% TeX-master: "main"
%%% End:

These previous results will allow to prove the simulation of $\calcPPC$ into
$\calcPPCX$ via the translation $\toPPCX{\_}$. We postpone this result to
Sec.~\ref{s:bisimulation} (\cf Thm.~\ref{t:bisimulation:ppc-ppcx}).

\medskip
Now we focus on the converse side of the embedding, \ie the translation of
$\calcPPCX$ terms into $\calcPPC$ terms. As before, this mapping requires two
lists of lists of symbols from which names of the free indices of the term
will be selected: one for variable indices and the other for matchable indices.

\begin{definition}
Given a term $t \in \Term{\calcPPCX}$ and lists of lists of distinct symbols
$V$ and $M$ such that $V_{ij}$ is defined for every $\termidx{i}[j] \in \fv{t}$
and $M_{ij}$ is defined for every $\termidx{i}[j] \in \fm{t}$, the
\emph{translation of $t$ relative to $V$ and $M$}, written $\toPPC[V][M]{t}$,
is inductively defined as follows: \[
\begin{array}{rcl@{\quad}l}
\toPPC[V][M]{\termidxv{i}[j]}    & \eqdef  & \termvar{V_{ij}} \\
\toPPC[V][M]{\termidxm{i}[j]}    & \eqdef  & \termmatch{M_{ij}} \\
\toPPC[V][M]{\termapp{t}{u}}     & \eqdef  & \termapp{\toPPC[V][M]{t}}{\toPPC[V][M]{u}} \\
\toPPC[V][M]{\termabs[n]{p}{t}}  & \eqdef  & \termabs[\theta]{\toPPC[V][\listconc{\theta}{M}]{p}}{\toPPC[\listconc{\theta}{V}][M]{t}}  & \text{$\theta = \lista{x_1, \ldots, x_n}$ fresh symbols}
\end{array}
\] Let $x_1, x_2, \ldots$ be the same enumeration of $\TermVariable$ as in
Def.~\ref{d:translation:to-ppcx}. Then, the \emph{translation} of $t$ to
$\calcPPC$, written simply $\toPPC{t}$, is defined as $\toPPC[X][X]{t}$ where
$X = \lista{\lista{x_1}, \ldots, \lista{x_n}}$ such that $\fv{t} \cup \fm{t}
\subseteq \set{\termidxv{1}[1], \ldots, \termidxv{n}[1]}$. Note that
well-formedness of terms guarantees that $X$ satisfies the conditions above.
\label{d:translation:to-ppc}
\end{definition}

To illustrate the translation, consider the term $t_1 =
\termapp{(\termabs[1]{\termapp{\termidxm{2}[1]}{\termidxm{1}[1]}}{\termidxv{1}[1]})}{(\termapp{\termidxm{1}[1]}{\termidxv{2}[1]})}$
where $\fv{t_1} = \set{\termidx{2}[1]}$ and $\fm{t_1} =
\set{\termidx{1}[1]}$. Then,
$\toPPC[\lista{\lista{y},\lista{z}}][\lista{\lista{y},\lista{z}}]{t_1} =
\termapp{(\termabs[\lista{x}]{\toPPC[\lista{\lista{y},\lista{z}}][\lista{\lista{x},\lista{y},\lista{z}}]{\termapp{\termidxm{2}[1]}{\termidxm{1}[1]}}}{\toPPC[\lista{\lista{x},\lista{y},\lista{z}}][\lista{\lista{y},\lista{z}}]{\termidxv{1}[1]}})}{\toPPC[\lista{\lista{y},\lista{z}}][\lista{\lista{y},\lista{z}}]{\termapp{\termidxm{1}[1]}{\termidxv{2}[1]}}}
=
\termapp{(\termabs[\lista{x}]{\termapp{\termmatch{y}}{\termmatch{x}}}{\termvar{x}})}{(\termapp{\termmatch{y}}{\termvar{z}})}$.
Note that $t_1 = \toPPCX{s_1}$ from the example after
Def.~\ref{d:translation:to-ppcx} and, with a proper initialisation of the lists
$V$ and $M$, we get $\toPPC{t_1} = s_1$.

\medskip
Once again, we start with some technical lemmas for substitutions and the
matching operations with respect to the embedding $\toPPC{\_}$. In this case,
the increment functions for variable and matchable indices behave as follows:

\begin{lemma}
Let $t \in \Term{\calcPPCX}$, $k \geq 0$, $i \geq 1$ and $n \geq k+i$ such
that $\fv{t} \cup \fm{t} \subseteq
\set{\termidx{i'}[j] \mid i' \leq n-(i-1), j \leq \fsize{X_{i'}}}$. Then,
\begin{enumerate}
  \item\label{l:translation:to-ppc-upd:var}
  $\toPPC[\listconc{\listconc{X_1}{\ldots}}{X_n}][M]{\fincv[k]{t}[i-1]} \eqalpha
  \toPPC[\listconc{\listconc{\listconc{X_1}{\ldots}}{X_k}}{\listconc{\listconc{X_{k+i}}{\ldots}}{X_n}}][M]{t}$.

  \item\label{l:translation:to-ppc-upd:match}
  $\toPPC[V][\listconc{\listconc{X_1}{\ldots}}{X_n}]{\fincm[k]{t}[i-1]} \eqalpha
  \toPPC[V][\listconc{\listconc{\listconc{X_1}{\ldots}}{X_k}}{\listconc{\listconc{X_{k+i}}{\ldots}}{X_n}}]{t}$.
\end{enumerate}
\label{l:translation:to-ppc-upd}
\end{lemma}

\begin{proof}
\begin{enumerate}
  \item By induction on $t$.
  \begin{itemize}
    \item $t = \termidxv{i'}[j]$. By hypothesis, $i' \leq n-(i-1)$. Then, there
    are two possible cases:
    \begin{enumerate}
      \item $i' \leq k$. Then, $\fincv[k]{\termidxv{i'}[j]}[i-1] =
      \termidxv{i'}[j]$. Thus,
      $\toPPC[\listconc{\listconc{\listconc{X_1}{\ldots}}{X_k}}{\listconc{\listconc{X_{k+i}}{\ldots}}{X_n}}][M]{\termidxv{i'}[j]}
      = X_{i'j} =
      \toPPC[\listconc{\listconc{X_1}{\ldots}}{X_n}][M]{\termidxv{i'}[j]} =
      \toPPC[\listconc{\listconc{X_1}{\ldots}}{X_n}][M]{\fincv[k]{\termidxv{i'}[j]}[i-1]}$.

      \item $i' > k$. Then, $\fincv[k]{\termidxv{i'}[j]}[i-1] =
      \termidxv{(i'+(i-1))}[j]$. Thus,
      $\toPPC[\listconc{\listconc{\listconc{X_1}{\ldots}}{X_k}}{\listconc{\listconc{X_{k+i}}{\ldots}}{X_n}}][M]{\termidxv{i'}[j]}
      = X_{(i'+(i-1))j} =
      \toPPC[\listconc{\listconc{X_1}{\ldots}}{X_n}][M]{\fincv[k]{\termidxv{i'}[j]}[i-1]}$.
    \end{enumerate}

    \item $t = \termidxm{i'}[j]$. This is immediate since
    $\toPPC[\listconc{\listconc{X_1}{\ldots}}{X_n}][M]{\termidxm{i'}[j]} =
    M_{i'j} =
    \toPPC[\listconc{\listconc{\listconc{X_1}{\ldots}}{X_k}}{\listconc{\listconc{X_{k+i}}{\ldots}}{X_n}}][M]{\termidxm{i'}[j]}$
    and $\fincv[k]{\_}[i-1]$ leaves matchable indices untouched.

    \item $t = \termapp{s}{u}$. By \ih
    $\toPPC[\listconc{\listconc{X_1}{\ldots}}{X_n}][M]{\fincv[k]{s}[i-1]}
    \eqalpha
    \toPPC[\listconc{\listconc{\listconc{X_1}{\ldots}}{X_k}}{\listconc{\listconc{X_{k+i}}{\ldots}}{X_n}}][M]{s}$
    and $\toPPC[\listconc{\listconc{X_1}{\ldots}}{X_n}][M]{\fincv[k]{u}[i-1]}
    \eqalpha
    \toPPC[\listconc{\listconc{\listconc{X_1}{\ldots}}{X_k}}{\listconc{\listconc{X_{k+i}}{\ldots}}{X_n}}][M]{u}$.
    Thus, we conclude
    $\toPPC[\listconc{\listconc{X_1}{\ldots}}{X_n}][M]{\fincv[k]{t}[i-1]} =
    \termapp{\toPPC[\listconc{\listconc{X_1}{\ldots}}{X_n}][M]{\fincv[k]{s}[i-1]}}{\toPPC[\listconc{\listconc{X_1}{\ldots}}{X_n}][M]{\fincv[k]{u}[i-1]}}
    =
    \termapp{\toPPC[\listconc{\listconc{\listconc{X_1}{\ldots}}{X_k}}{\listconc{\listconc{X_{k+i}}{\ldots}}{X_n}}][M]{s}}{\toPPC[\listconc{\listconc{\listconc{X_1}{\ldots}}{X_k}}{\listconc{\listconc{X_{k+i}}{\ldots}}{X_n}}][M]{u}}
    =
    \toPPC[\listconc{\listconc{\listconc{X_1}{\ldots}}{X_k}}{\listconc{\listconc{X_{k+i}}{\ldots}}{X_n}}][M]{t}$.

    \item $t = \termabs[n]{p}{s}$. Let $\theta$ be a list of $n$ fresh symbols.
    By \ih
    $\toPPC[\listconc{\listconc{X_1}{\ldots}}{X_n}][\listconc{\theta}{M}]{\fincv[k]{p}[i-1]}
    \eqalpha
    \toPPC[\listconc{\listconc{\listconc{X_1}{\ldots}}{X_k}}{\listconc{\listconc{X_{k+i}}{\ldots}}{X_n}}][\listconc{\theta}{M}]{p}$
    and $\toPPC[\listconc{\theta}{\listconc{\listconc{X_1}{\ldots}}{X_n}}][M]{\fincv[k+1]{s}[i-1]}
    \eqalpha
    \toPPC[\listconc{\theta}{\listconc{\listconc{\listconc{X_1}{\ldots}}{X_k}}{\listconc{\listconc{X_{k+i}}{\ldots}}{X_n}}}][M]{s}$.
    Note that $\theta$ is pushed accordingly in the list of variable or
    matchable symbols for $s$ and $p$ respectively, following the definition of
    $\toPPC{\_}$ for abstractions. In the case for $s$, this implies concluding
    with $\fincv[k+1]{\_}$ instead of $\fincv[k]{\_}$. Finally, we conclude by
    Def.~\ref{d:translation:to-ppc},
    $\toPPC[\listconc{\listconc{X_1}{\ldots}}{X_n}][M]{\fincv[k]{t}[i-1]}\ = \ 
    \termabs[\theta]{\toPPC[\listconc{\listconc{X_1}{\ldots}}{X_n}][\listconc{\theta}{M}]{\fincv[k]{p}[i-1]}}{\toPPC[\listconc{\listconc{\listconc{\theta}{X_1}}{\ldots}}{X_n}][M]{\fincv[k]{s}[i]}}
    \ \eqalpha \ 
    \termabs[\theta]{\toPPC[\listconc{\listconc{\listconc{X_1}{\ldots}}{X_k}}{\listconc{\listconc{X_{k+i}}{\ldots}}{X_n}}][\listconc{\theta}{M}]{p}}{\toPPC[\listconc{\theta}{\listconc{\listconc{\listconc{X_1}{\ldots}}{X_k}}{\listconc{\listconc{X_{k+i}}{\ldots}}{X_n}}}][M]{s}}
    \ = \ 
    \toPPC[\listconc{\listconc{\listconc{X_1}{\ldots}}{X_k}}{\listconc{\listconc{X_{k+i}}{\ldots}}{X_n}}][M]{t}$.
  \end{itemize}

  \item By induction on $t$. This item is similar to the previous one.
\end{enumerate}
\end{proof}

%%% Local Variables:
%%% mode: latex
%%% TeX-master: "main"
%%% End:

As for the converse, the translation is only defined for substitution at level
$1$ and requires to be provided a list of symbols $\theta$ such that
$\fsize{\theta} \geq \max\set{j \mid \termidx{1}[j] \in \dom{\sigma}}$. Then,
$\toPPC[V][M]{\sigma}[\theta] \eqdef
\subs{\theta_j}{\toPPC[V][M]{\subsapply{\sigma}{\termidxv{1}[j]}}}_{\termidxv{1}[j] \in \dom{\sigma}}$.
The application of a substitution at an arbitrary level $i$ is shown to
translate properly.

\begin{lemma}
Let $s \in \Term{\calcPPCX}$, $\sigma$ be a substitution at level $i$,
$\theta$ be a list of fresh symbols such that $\fsize{\theta} \geq
\max\set{j \mid \termidx{i}[j] \in \dom{\sigma}}$ and $Y$ be a list of $i-1$
lists of symbols. Then,
$\toPPC[\listconc{Y}{X}][M]{\fdecv[i-1]{\subsapply{\subs{\termidx{i}[j]}{\fincv[i-1]{\subsapply{\sigma}{\termidx{i}[j]}}}_{\termidx{i}[j] \in \dom{\sigma}}}{s}}}
\eqalpha
\subsapply{\subs{\theta_j}{\toPPC[\listconc{Y}{X}][M]{\subsapply{\sigma}{\termidxv{i}[j]}}}_{\termidx{i}[j] \in \dom{\sigma}}}{\toPPC[\listconc{\listconc{Y}{\theta}}{X}][M]{s}}$.
\label{l:translation:to-ppc-subs}
\end{lemma}

\begin{proof}
By induction on term $s$.
\begin{itemize}
  \item $s = \termidxv{i'}[k]$. There are three possible cases.
  \begin{enumerate}
    \item $i' = i$ and $k \leq \fsize{\theta}$. Then,
    $\toPPC[\listconc{\listconc{Y}{\theta}}{X}][M]{\termidxv{i}[k]} =
    \termvar{\theta_k}$. By definition, we have
    $\subsapply{\subs{\termvar{\theta_j}}{\toPPC[\listconc{Y}{X}][M]{\subsapply{\sigma}{\termidxv{i}[j]}}}_{\termidxv{i}[j] \in \dom{\sigma}}}{\toPPC[\listconc{\listconc{Y}{\theta}}{X}][M]{\termidxv{i}[k]}}
    =
    \subsapply{\subs{\termvar{\theta_j}}{\toPPC[\listconc{Y}{X}][M]{\subsapply{\sigma}{\termidxv{i}[j]}}}_{\termidxv{i}[j] \in \dom{\sigma}}}{\termvar{\theta_k}}
    =
    \toPPC[\listconc{Y}{X}][M]{\subsapply{\sigma}{\termidxv{i}[k]}}$. Moreover,
    by Lem.~\ref{l:appendix:finc} (\ref{l:appendix:finc:fdec-l}),
    $\subsapply{\sigma}{\termidxv{i}[k]} =
    \fdecv[i-1]{\fincv[i-1]{\subsapply{\sigma}{\termidxv{i}[k]}}}$ and, by
    definition, $\fincv[i-1]{\subsapply{\sigma}{\termidxv{i}[k]}}$ = 
    $\subsapply{\subs{\termidxv{i}[j]}{\fincv[i-1]{\subsapply{\sigma}{\termidxv{i}[j]}}}_{\termidxv{i}[j] \in \dom{\sigma}}}{\termidxv{i}[k]}$.
    Hence, $\subsapply{\subs{\termvar{\theta_j}}{\toPPC[\listconc{Y}{X}][M]{\subsapply{\sigma}{\termidxv{i}[j]}}}_{\termidxv{i}[j] \in \dom{\sigma}}}{\toPPC[\listconc{\listconc{Y}{\theta}}{X}][M]{\termidxv{i}[k]}}
    = \toPPC[\listconc{Y}{X}][M]{\subsapply{\sigma}{\termidxv{i}[k]}} =
    \toPPC[\listconc{Y}{X}][M]{\fdecv[i-1]{\subsapply{\subs{\termidxv{i}[j]}{\fincv[i-1]{\subsapply{\sigma}{\termidxv{i}[j]}}}_{\termidxv{i}[j] \in \dom{\sigma}}}{\termidxv{i}[k]}}}$ and we conclude.
    
    \item $i' < i$. Then,
    $\toPPC[\listconc{\listconc{Y}{\theta}}{X}][M]{\termidxv{i'}[k]} =
    \termvar{Y_{i'k}}$. By definition, we have
    $\subsapply{\subs{\termvar{\theta_j}}{\toPPC[\listconc{Y}{X}][M]{\subsapply{\sigma}{\termidxv{i}[j]}}}_{\termidxv{i}[j] \in \dom{\sigma}}}{\toPPC[\listconc{\listconc{Y}{\theta}}{X}][M]{\termidxv{i'}[k]}}
    = \toPPC[\listconc{\listconc{Y}{\theta}}{X}][M]{\termidxv{i'}[k]}$.
    Moverover, since $i' \neq i$, by Lem.~\ref{l:appendix:finc}
    (\ref{l:appendix:finc:fdec-r}), we get
    $\toPPC[\listconc{\listconc{Y}{\theta}}{X}][M]{\termidxv{i'}[k]} =
    \toPPC[\listconc{\listconc{Y}{\theta}}{X}][M]{\fincv[i-1]{\fdecv[i-1]{\termidxv{i'}[k]}}}$.
    Then, by Lem.~\ref{l:translation:to-ppc-upd}
    (\ref{l:translation:to-ppc-upd:var}),
    $\toPPC[\listconc{\listconc{Y}{\theta}}{X}][M]{\fincv[i-1]{\fdecv[i-1]{\termidxv{i'}[k]}}}
    = \toPPC[\listconc{Y}{X}][M]{\fdecv[i-1]{\termidxv{i'}[k]}}$.
    Thus, $\toPPC[\listconc{Y}{X}][M]{\fdecv[i-1]{\termidxv{i'}[k]}} =
    \toPPC[\listconc{Y}{X}][M]{\fdecv[i-1]{\subsapply{\subs{\termidxv{i}[j]}{\fincv[i-1]{\subsapply{\sigma}{\termidxv{i}[j]}}}_{\termidxv{i}[j] \in \dom{\sigma}}}{\termidxv{i'}[k]}}}$
    and we conclude given that $\termidxv{i'}[k]$ is not affected by the
    substitution.
    
    \item $i' > i$. Then,
    $\toPPC[\listconc{\listconc{Y}{\theta}}{X}][M]{\termidxv{i'}[k]} =
    \termvar{X_{i'k}}$. By definition, we have
    $\subsapply{\subs{\termvar{\theta_j}}{\toPPC[\listconc{Y}{X}][M]{\subsapply{\sigma}{\termidxv{i}[j]}}}_{\termidxv{i}[j] \in \dom{\sigma}}}{\toPPC[\listconc{\listconc{Y}{\theta}}{X}][M]{\termidxv{i'}[k]}}
    = \toPPC[\listconc{\listconc{Y}{\theta}}{X}][M]{\termidxv{i'}[k]}$.
    Moverover, since $i' \neq i$, by Lem.~\ref{l:appendix:finc}
    (\ref{l:appendix:finc:fdec-r}), we get
    $\toPPC[\listconc{\listconc{Y}{\theta}}{X}][M]{\termidxv{i'}[k]} =
    \toPPC[\listconc{\listconc{Y}{\theta}}{X}][M]{\fincv[i-1]{\fdecv[i-1]{\termidxv{i'}[k]}}}$.
    Then, by Lem.~\ref{l:translation:to-ppc-upd}
    (\ref{l:translation:to-ppc-upd:var}),
    $\toPPC[\listconc{\listconc{Y}{\theta}}{X}][M]{\fincv[i-1]{\fdecv[i-1]{\termidxv{i'}[k]}}}
    = \toPPC[\listconc{Y}{X}][M]{\fdecv[i-1]{\termidxv{i'}[k]}}$.
    Thus, $\toPPC[\listconc{Y}{X}][M]{\fdecv[i-1]{\termidxv{i'}[k]}} =
    \toPPC[\listconc{Y}{X}][M]{\fdecv[i-1]{\subsapply{\subs{\termidxv{i}[j]}{\fincv[i-1]{\subsapply{\sigma}{\termidxv{i}[j]}}}_{\termidxv{i}[j] \in \dom{\sigma}}}{\termidxv{i'}[k]}}}$
    and we conclude given that $\termidxv{i'}[k]$ is not affected by the
    substitution.
  \end{enumerate}
  
  \item $s = \termidxm{i'}[k]$. Then,
  $\toPPC[\listconc{\listconc{Y}{\theta}}{X}][M]{\termidxm{i'}[k]} =
  \termmatch{M_{i'k}}$. Since matchables are not affected by the substitution,
  it is safe to conclude
  $\subsapply{\subs{\termvar{\theta_j}}{\toPPC[\listconc{Y}{X}][M]{\subsapply{\sigma}{\termidxm{i}[j]}}}_{\termidxv{i}[j] \in \dom{\sigma}}}{\toPPC[\listconc{\listconc{Y}{\theta}}{X}][M]{\termidxm{i'}[k]}}
  =
  \toPPC[\listconc{Y}{X}][M]{\fdecv[i-1]{\subsapply{\subs{\termidxv{i}[j]}{\fincv[i-1]{\subsapply{\sigma}{\termidxv{i}[j]}}}_{\termidxv{i}[j] \in \dom{\sigma}}}{\termidxm{i'}[k]}}}$.
  
  \item $s = \termapp{t}{u}$. This case is immediate from the \ih since every
  definition involved distributes over applications.
  
  \item $s = \termabs[m]{p}{t}$. Let $\theta'$ be a list of $m$ fresh symbols.
  Then, $\toPPC[\listconc{\listconc{Y}{\theta}}{X}][M]{s} \eqalpha
  \termabs[\theta']{\toPPC[\listconc{\listconc{Y}{\theta}}{X}][\listconc{\theta'}{M}]{p}}{\toPPC[\listconc{\listconc{\listconc{\theta'}{Y}}{\theta}}{X}][M]{t}}$.
  Moreover, by definition we get
  $\subsapply{\subs{\termvar{\theta_j}}{\toPPC[\listconc{Y}{X}][M]{\subsapply{\sigma}{\termidxm{i}[j]}}}_{\termidxv{i}[j] \in \dom{\sigma}}}{\toPPC[\listconc{\listconc{Y}{\theta}}{X}][M]{s}}
  \eqalpha
  \termabs[\theta']{\subsapply{\subs{\termvar{\theta_j}}{\toPPC[\listconc{Y}{X}][M]{\subsapply{\sigma}{\termidxm{i}[j]}}}_{\termidxv{i}[j] \in \dom{\sigma}}}{\toPPC[\listconc{\listconc{Y}{\theta}}{X}][\listconc{\theta'}{M}]{p}}}{\subsapply{\subs{\termvar{\theta_j}}{\toPPC[\listconc{Y}{X}][M]{\subsapply{\sigma}{\termidxm{i}[j]}}}_{\termidxv{i}[j] \in \dom{\sigma}}}{\toPPC[\listconc{\listconc{\listconc{\theta'}{Y}}{\theta}}{X}][M]{t}}}$.
  By Lem.~\ref{l:translation:to-ppc-upd} (\ref{l:translation:to-ppc-upd:var})
  and (\ref{l:translation:to-ppc-upd:match}), we have
  $\toPPC[\listconc{Y}{X}][M]{\subsapply{\sigma}{\termidxm{i}[j]}} =
  \toPPC[\listconc{\listconc{\theta'}{Y}}{X}][M]{\fincv{\subsapply{\sigma}{\termidxm{i}[j]}}}$
  and $\toPPC[\listconc{Y}{X}][M]{\subsapply{\sigma}{\termidxm{i}[j]}} =
  \toPPC[\listconc{Y}{X}][\listconc{\theta'}{M}]{\fincm{\subsapply{\sigma}{\termidxm{i}[j]}}}$
  respectively. Then, applying the \ih we get
  $\subsapply{\subs{\termvar{\theta_j}}{\toPPC[\listconc{Y}{X}][\listconc{\theta'}{M}]{\fincm{\subsapply{\sigma}{\termidxm{i}[j]}}}}_{\termidxv{i}[j] \in \dom{\sigma}}}{\toPPC[\listconc{\listconc{Y}{\theta}}{X}][\listconc{\theta'}{M}]{p}}
  \eqalpha
  \toPPC[\listconc{Y}{X}][\listconc{\theta'}{M}]{\fdecv[i-1]{\subsapply{\subs{\termidxv{i}[j]}{\fincv[i-1]{\fincm{\subsapply{\sigma}{\termidxv{i}[j]}}}}_{\termidxv{i}[j] \in \dom{\sigma}}}{p}}}$
  and 
  $\subsapply{\subs{\termvar{\theta_j}}{\toPPC[\listconc{\listconc{\theta'}{Y}}{X}][M]{\fincv{\subsapply{\sigma}{\termidxm{i}[j]}}}}_{\termidxv{i}[j] \in \dom{\sigma}}}{\toPPC[\listconc{\listconc{\listconc{\theta'}{Y}}{\theta}}{X}][M]{t}}
  \eqalpha
  \toPPC[\listconc{\listconc{\theta'}{Y}}{X}][M]{\fdecv[i]{\subsapply{\subs{\termidxv{(i+1)}[j]}{\fincv[i]{\fincv{\subsapply{\sigma}{\termidxv{i}[j]}}}}_{\termidxv{i}[j] \in \dom{\sigma}}}{t}}}$.
  Furthermore, by Lem.~\ref{l:appendix:finc} (\ref{l:appendix:finc:match}),
  $\fincv[i-1]{\fincm{\subsapply{\sigma}{\termidxv{i}[j]}}} =
  \fincm{\fincv[i-1]{\subsapply{\sigma}{\termidxv{i}[j]}}}$ and, by
  Lem.~\ref{l:appendix:finc} (\ref{l:appendix:finc:var}),
  $\fincv[i]{\fincv{\subsapply{\sigma}{\termidxv{i}[j]}}} =
  \fincv{\fincv[i-1]{\subsapply{\sigma}{\termidxv{i}[j]}}}$ since $i > 0$ by
  definition. Finally, we conclude as follows: \[%\kern-6em
\begin{array}{ll}
          & \subsapply{\subs{\termvar{\theta_j}}{\toPPC[\listconc{Y}{X}][M]{\subsapply{\sigma}{\termidxm{i}[j]}}}_{\termidxv{i}[j] \in \dom{\sigma}}}{\toPPC[\listconc{\listconc{Y}{\theta}}{X}][M]{\termabs[m]{p}{t}}} \\
\eqalpha  & \termabs[\theta']{\subsapply{\subs{\termvar{\theta_j}}{\toPPC[\listconc{Y}{X}][\listconc{\theta'}{M}]{\fincm{\subsapply{\sigma}{\termidxm{i}[j]}}}}_{\termidxv{i}[j] \in \dom{\sigma}}}{\toPPC[\listconc{\listconc{Y}{\theta}}{X}][\listconc{\theta'}{M}]{p}}}{\subsapply{\subs{\termvar{\theta_j}}{\toPPC[\listconc{\listconc{\theta'}{Y}}{X}][M]{\fincv{\subsapply{\sigma}{\termidxm{i}[j]}}}}_{\termidxv{i}[j] \in \dom{\sigma}}}{\toPPC[\listconc{\listconc{\listconc{\theta'}{Y}}{\theta}}{X}][M]{t}}} \\
\eqalpha  & \termabs[\theta']{\toPPC[\listconc{Y}{X}][\listconc{\theta'}{M}]{\fdecv[i-1]{\subsapply{\subs{\termidxv{i}[j]}{\fincv[i-1]{\fincm{\subsapply{\sigma}{\termidxv{i}[j]}}}}_{\termidxv{i}[j] \in \dom{\sigma}}}{p}}}}{\toPPC[\listconc{\listconc{\theta'}{Y}}{X}][M]{\fdecv[i]{\subsapply{\subs{\termidxv{(i+1)}[j]}{\fincv[i]{\fincv{\subsapply{\sigma}{\termidxv{i}[j]}}}}_{\termidxv{i}[j] \in \dom{\sigma}}}{t}}}} \\
=         & \termabs[\theta']{\toPPC[\listconc{Y}{X}][\listconc{\theta'}{M}]{\fdecv[i-1]{\subsapply{\subs{\termidxv{i}[j]}{\fincm{\fincv[i-1]{\subsapply{\sigma}{\termidxv{i}[j]}}}}_{\termidxv{i}[j] \in \dom{\sigma}}}{p}}}}{\toPPC[\listconc{\listconc{\theta'}{Y}}{X}][M]{\fdecv[i]{\subsapply{\subs{\termidxv{(i+1)}[j]}{\fincv{\fincv[i-1]{\subsapply{\sigma}{\termidxv{i}[j]}}}}_{\termidxv{i}[j] \in \dom{\sigma}}}{t}}}} \\
=         & \toPPC[\listconc{Y}{X}][M]{\termabs[m]{\fdecv[i-1]{\subsapply{\subs{\termidxv{i}[j]}{\fincm{\fincv[i-1]{\subsapply{\sigma}{\termidxv{i}[j]}}}}_{\termidxv{i}[j] \in \dom{\sigma}}}{p}}}{\fdecv[i]{\subsapply{\subs{\termidxv{(i+1)}[j]}{\fincv{\fincv[i-1]{\subsapply{\sigma}{\termidxv{i}[j]}}}}_{\termidxv{i}[j] \in \dom{\sigma}}}{t}}}} \\
=         & \toPPC[\listconc{Y}{X}][M]{\fdecv[i-1]{\termabs[m]{\subsapply{\subs{\termidxv{i}[j]}{\fincm{\fincv[i-1]{\subsapply{\sigma}{\termidxv{i}[j]}}}}_{\termidxv{i}[j] \in \dom{\sigma}}}{p}}{\subsapply{\subs{\termidxv{(i+1)}[j]}{\fincv{\fincv[i-1]{\subsapply{\sigma}{\termidxv{i}[j]}}}}_{\termidxv{i}[j] \in \dom{\sigma}}}{t}}}} \\
=         & \toPPC[\listconc{Y}{X}][M]{\fdecv[i-1]{\subsapply{\subs{\termidxv{i}[j]}{\fincv[i-1]{\subsapply{\sigma}{\termidxv{i}[j]}}}_{\termidxv{i}[j] \in \dom{\sigma}}}{\termabs[m]{p}{t}}}} \\
\end{array} \]
\end{itemize}
\end{proof}

%%% Local Variables:
%%% mode: latex
%%% TeX-master: "main"
%%% End:

Similarly to the substitution case, the translation of a match
$\match[n]{p}{u}$ requires to be supplied with a list of $n$ fresh symbols
$\theta$. Then, it is defined as $\toPPC[V][M]{\match[n]{p}{u}}[\theta] \eqdef
\match[\theta]{\toPPC[V][\listconc{\theta}{M}]{p}}{\toPPC[V][M]{u}}$. The newly
provided list of symbols is used both as the parameter of the resulting match
and to properly translate the pattern, obtaining the following expected
results.

\begin{lemma}
Let $t \in \Term{\calcPPCX}$. Then,
\begin{enumerate}
  \item \label{l:appendix:to-ppc-mf:data} $t \in \TermData{\calcPPCX}$ iff
  $\toPPC[V][M]{t} \in \TermData{\calcPPC}$.
  \item \label{l:appendix:to-ppc-mf:match} $t \in \Matchable{\calcPPCX}$ iff
  $\toPPC[V][M]{t} \in \Matchable{\calcPPC}$.
\end{enumerate}
\label{l:appendix:to-ppc-mf}
\end{lemma}

\begin{proof}
Both items follow by straightforward induction on $t$, using
(\ref{l:appendix:to-ppc-mf:data}) to prove (\ref{l:appendix:to-ppc-mf:match}).
\end{proof}

\begin{lemma}
Let $p, u \in \Term{\calcPPCX}$.
\begin{enumerate}
  \item\label{l:translation:to-ppc-match:subs} If $\match[n]{p}{u} =
  \sigma$, then $\toPPC[V][M]{\match[n]{p}{u}}[\theta] =
  \toPPC[V][M]{\sigma}[\theta]$.

  \item\label{l:translation:to-ppc-match:fail} If $\match[n]{p}{u} =
  \matchfail$, then $\toPPC[V][M]{\match[n]{p}{u}}[\theta] =
  \matchfail$.

  \item\label{l:translation:to-ppc-match:wait} If $\match[n]{p}{u} =
  \matchundet$, then $\toPPC[V][M]{\match[n]{p}{u}}[\theta] =
  \matchundet$.
\end{enumerate}
\label{l:translation:to-ppc-match}
\end{lemma}

\begin{proof}
By straightforward induction on $p$ using Lem.~\ref{l:appendix:to-ppc-mf}.
Similar to the proof for Lem.~\ref{l:translation:to-ppcx-match}.
\end{proof}

Now we are in conditions to prove the simulation of $\calcPPCX$ into $\calcPPC$
via the translation provided in Def.~\ref{d:translation:to-ppc}.

\medskip
Before proceeding to the next section, one final result concerns the
translations. It turns out that each translation is the inverse of the other,
as shown in Thm.~\ref{t:translation:inverse}. In case of $\calcPPC$ terms
we should work modulo $\alpha$-conversion, while for $\calcPPCX$ terms we may
use equality (modulo secondary indices permutations, \cf last paragraph in
Sec.~\ref{s:ppcx}). This constitutes the main result of this section and is the
key to extend our individual simulation results (\cf
Thm.~\ref{t:bisimulation:ppc-ppcx} and~\ref{t:bisimulation:ppcx-ppc} resp.)
into a strong bisimulation between the two calculi, as shown in
Sec.~\ref{s:bisimulation}.

\begin{lemma}
Let $t \in \Term{\calcPPCX}$. Then, $\toPPCX[V][M]{\toPPC[V][M]{t}} = t$.
\label{l:translation:inverse-ppcx}
\end{lemma}

\begin{proof}
By induction on $t$, assuming $V$ and $M$ both satisfy the conditions of
Def.~\ref{d:translation:to-ppcx} and Def.~\ref{d:translation:to-ppc}.
\begin{itemize}
  \item $t = \termidxv{i}[j]$. By Def.~\ref{d:translation:to-ppc},
  $\toPPC[V][M]{\termidxv{i}[j]} = \termvar{V_{ij}}$. Moreover, since all the
  symbols in $V$ are distinct by hypothesis, we have $i =
  \min\set{i' \mid V_{ij} \in V_{i'}}$ and $j = \min\set{j' \mid V_{ij} = V_{ij'}}$.
  Then, we conclude by Def.~\ref{d:translation:to-ppcx},
  $\toPPCX[V][M]{\toPPC[V][M]{\termidxv{i}[j]}} =
  \toPPCX[V][M]{\termvar{V_{ij}}} = \termidxv{i}[j]$.
  
  \item $t = \termidxm{i}[j]$. By Def.~\ref{d:translation:to-ppc},
  $\toPPC[V][M]{\termidxm{i}[j]} = \termmatch{M_{ij}}$. Moreover, since all the
  symbols in $M$ are distinct by hypothesis, we have
  $i = \min\set{i' \mid M_{ij} \in M_{i'}}$ and $j = \min\set{j' \mid M_{ij} = M_{ij'}}$.
  Then, we conclude by Def.~\ref{d:translation:to-ppcx},
  $\toPPCX[V][M]{\toPPC[V][M]{\termidxm{i}[j]}} =
  \toPPCX[V][M]{\termmatch{M_{ij}}} = \termidxm{i}[j]$.
  
  \item $t = \termapp{s}{u}$. By \ih we have $\toPPCX[V][M]{\toPPC[V][M]{s}} =
  s$ and $\toPPCX[V][M]{\toPPC[V][M]{u}} = u$. Then, we conclude by
  Def.~\ref{d:translation:to-ppcx} and Def.~\ref{d:translation:to-ppc},
  $\toPPCX[V][M]{\toPPC[V][M]{t}} =
  \termapp{\toPPCX[V][M]{\toPPC[V][M]{s}}}{\toPPCX[V][M]{\toPPC[V][M]{u}}} =
  \termapp{s}{u} = t$.
  
  \item $t = \termabs[n]{p}{s}$. By Def.~\ref{d:translation:to-ppc},
  $\toPPC[V][M]{t} =
  \termabs[\theta]{\toPPC[V][\listconc{\theta}{M}]{p}}{\toPPC[\listconc{\theta}{V}][M]{s}}$
  with $\theta$ a list of $n$ fresh symbols, \ie $\fsize{\theta} = n$.
  Moreover, by Def.~\ref{d:translation:to-ppcx}, $\toPPCX[V][M]{\toPPC[V][M]{t}} =
  \termabs[n]{\toPPCX[V][\listconc{\theta}{M}]{\toPPC[V][\listconc{\theta}{M}]{p}}}{\toPPCX[\listconc{\theta}{V}][M]{\toPPC[\listconc{\theta}{V}][M]{s}}}$.
  By \ih we have
  $\toPPCX[V][\listconc{\theta}{M}]{\toPPC[V][\listconc{\theta}{M}]{p}} = p$
  and $\toPPCX[\listconc{\theta}{V}][M]{\toPPC[\listconc{\theta}{V}][M]{s}} =
  s$. Thus, we conclude $\toPPCX[V][M]{\toPPC[V][M]{t}} = \termabs[n]{p}{s}$.
\end{itemize}
\end{proof}

%%% Local Variables: 
%%% mode: latex
%%% TeX-master: "main"
%%% End: 

\begin{lemma}
Let $t \in \Term{\calcPPC}$. Then, $\toPPC[V][M]{\toPPCX[V][M]{t}} \eqalpha t$.
\label{l:translation:inverse-ppc}
\end{lemma}

\begin{proof}
By induction on $t$, assuming $V$ and $M$ both satisfy the conditions of
Def.~\ref{d:translation:to-ppcx} and Def.~\ref{d:translation:to-ppc}.
\begin{itemize}
  \item $t = \termvar{x}$. By Def.~\ref{d:translation:to-ppcx},
  $\toPPCX[V][M]{\termvar{x}} = \termidxv{i}[j]$ where $i =
  \min\set{i' \mid x \in V_{i'}}$ and $j = \min\set{j' \mid x = V_{ij'}}$.
  Then, we conclude by Def.~\ref{d:translation:to-ppc},
  $\toPPC[V][M]{\toPPCX[V][M]{\termvar{x}}} = \toPPC[V][M]{\termidxv{i}[j]} =
  \termvar{V_{ij}} = \termvar{x}$.
  
  \item $t = \termmatch{x}$. By Def.~\ref{d:translation:to-ppcx},
  $\toPPCX[V][M]{\termmatch{x}} = \termidxm{i}[j]$ where $i =
  \min\set{i' \mid x \in M_{i'}}$ and $j = \min\set{j' \mid x = M_{ij'}}$.
  Then, we conclude by Def.~\ref{d:translation:to-ppc},
  $\toPPC[V][M]{\toPPCX[V][M]{\termmatch{x}}} = \toPPC[V][M]{\termidxm{i}[j]} =
  \termmatch{M_{ij}} = \termmatch{x}$.
  
  \item $t = \termapp{s}{u}$. By \ih we have $\toPPC[V][M]{\toPPCX[V][M]{s}}
  \eqalpha s$ and $\toPPC[V][M]{\toPPCX[V][M]{u}} \eqalpha u$. Then, we
  conclude by Def.~\ref{d:translation:to-ppcx} and
  Def.~\ref{d:translation:to-ppc}, $\toPPC[V][M]{\toPPCX[V][M]{t}} =
  \termapp{\toPPC[V][M]{\toPPCX[V][M]{s}}}{\toPPC[V][M]{\toPPCX[V][M]{u}}}
  \eqalpha \termapp{s}{u} = t$.
  
  \item $t = \termabs[\theta]{p}{s}$. By Def.~\ref{d:translation:to-ppcx},
  $\toPPCX[V][M]{t} =
  \termabs[\fsize{\theta}]{\toPPCX[V][\listconc{\theta}{M}]{p}}{\toPPCX[\listconc{\theta}{V}][M]{s}}$.
  Note that $\theta$ is fresh for $V$ and $M$. Then, by $\alpha$-conversion and
  Def.~\ref{d:translation:to-ppc}, $\toPPC[V][M]{\toPPCX[V][M]{t}} \eqalpha
  \termabs[\theta]{\toPPC[V][\listconc{\theta}{M}]{\toPPCX[V][\listconc{\theta}{M}]{p}}}{\toPPC[\listconc{\theta}{V}][M]{\toPPCX[\listconc{\theta}{V}][M]{s}}}$.
  By \ih we have
  $\toPPC[V][\listconc{\theta}{M}]{\toPPCX[V][\listconc{\theta}{M}]{p}} \eqalpha
  p$ and $\toPPC[\listconc{\theta}{V}][M]{\toPPCX[\listconc{\theta}{V}][M]{s}}
  \eqalpha s$. Thus, we conclude $\toPPC[V][M]{\toPPCX[V][M]{t}} \eqalpha
  \termabs[\theta]{p}{s}$.
\end{itemize}
\end{proof}

%%% Local Variables: 
%%% mode: latex
%%% TeX-master: "main"
%%% End: 

\begin{theorem}[Invertibility]
Let $t \in \Term{\calcPPCX}$ and $s \in \Term{\calcPPC}$. Then,
\begin{inparaenum}
  \item \label{t:translation:inverse:ppcx} $\toPPCX{\toPPC{t}} = t$; and
  \item \label{t:translation:inverse:ppc} $\toPPC{\toPPCX{s}} \eqalpha s$.
\end{inparaenum}
\label{t:translation:inverse}
\end{theorem}

\begin{proof}
Both items are immediate by Lem.~\ref{l:translation:inverse-ppcx}
and~\ref{l:translation:inverse-ppc} respectively, taking $V = M = X$ as given
in Def.~\ref{d:translation:to-ppcx} and Def.~\ref{d:translation:to-ppc}.
\end{proof}

%%% Local Variables:
%%% mode: latex
%%% TeX-master: "main"
%%% End:

%%%%%%%%%%%%%%%%%%%%%%%%%%%%%%%%%%%%%%%%%%%%%%%%%%%%%%%%%%%%%%%%%%%%%%%%%%%%%%%
\section{Strong bisimulation}
\label{s:bisimulation}
%%%%%%%%%%%%%%%%%%%%%%%%%%%%%%%%%%%%%%%%%%%%%%%%%%%%%%%%%%%%%%%%%%%%%%%%%%%%%%%

In this section we prove the simulation of one calculus by the other via the
proper translation and, most importantly, the strong bisimulation that follows
after the invertibility result (\cf Thm.~\ref{t:translation:inverse}). This
strong bisimulation result will allow to port many important properties already
known for $\calcPPC$ into $\calcPPCX$, as we will discuss later.

\medskip
We start by extending the increment function to substitutions and matches. The
\emph{increment at depth $k$} for variable indices in a substitution at level
$1$ is defined as $\fincv[k]{\sigma} \eqdef
\subs{\termidxv{1}[j]}{\fincv[k]{\subsapply{\sigma}{\termidxv{1}[j]}}}_{\termidx{1}[j] \in \dom{\sigma}}$.
As for matches, we simply define $\fincv[k]{\matchfail} \eqdef \matchfail$ and
$\fincv[k]{\matchundet} \eqdef \matchundet$, and use the definition over
substitution for successful matches. Recall that the matching operation yields
only substitutions at level 1. Then, the following result holds.

\begin{lemma}
Let $p,u \in \Term{\calcPPCX}$. Then, $\match[n]{p}{\fincv{u}} =
\fincv{\match[n]{p}{u}}$.
\label{l:bisimulation:fincv-lifting}
\end{lemma}

\begin{proof}
By straightforward induction on $p$.
\end{proof}

Let us focus first on simulating $\calcPPC$ by $\calcPPCX$. The key step here
is the preservation of the matching operation shown for $\toPPCX{\_}$ in
Lem.~\ref{l:translation:to-ppcx-match}. It guarantees that every redex in
$\calcPPC$ turns into a redex in $\calcPPCX$ too. Then, the appropiate
definition of the operational semantics given for $\calcPPCX$ in
Sec.~\ref{s:ppcx} allows us to conclude.

\begin{lemma}
Let $t \in \Term{\calcPPC}$. If $t \reduce[\rPPC] t'$, then $\toPPCX[V][M]{t}
\reduce[\rPPCX] \toPPCX[V][M]{t'}$.
\label{l:bisimulation:ppc-ppcx}
\end{lemma}

\begin{proof}
By definition $t \reduce[\rPPC] t'$ implies $t =
\ctxtapply{\ctxt{C}}{\termapp{(\termabs[\theta]{p}{s})}{u}}$ and $t' =
\ctxtapply{\ctxt{C}}{\matchapply{\match[\theta]{p}{u}}{s}}$ with
$\match[\theta]{p}{u}$ decided. We proceed by induction on $\ctxt{C}$.
\begin{itemize}
  \item $\ctxt{C} = \Box$. Then, $t = \termapp{(\termabs[\theta]{p}{s})}{u}$
  and $t' = \matchapply{\match[\theta]{p}{u}}{s}$ with $\match[\theta]{p}{u}$ a
  decided match. Moreover, $\toPPCX[V][M]{t} =
  \termapp{(\termabs[\fsize{\theta}]{\toPPCX[V][\listconc{\theta}{M}]{p}}{\toPPCX[\listconc{\theta}{V}][M]{s}})}{\toPPCX[V][M]{u}}$.
  There are two possible cases:
  \begin{enumerate}
    \item $\match[\theta]{p}{u} = \sigma$. Then, $t' =
    \subsapply{\sigma}{s}$. By Lem.~\ref{l:translation:to-ppcx-match}
    (\ref{l:translation:to-ppcx-match:subs}),
    $\toPPCX[V][M]{\match[\theta]{p}{u}} =
    \match[\fsize{\theta}]{\toPPCX[V][\listconc{\theta}{M}]{p}}{\toPPCX[V][M]{u}}
    = \toPPCX[V][M]{\sigma}[\theta]$. Moreover, by
    Lem.~\ref{l:bisimulation:fincv-lifting},
    $\match[\fsize{\theta}]{\toPPCX[V][\listconc{\theta}{M}]{p}}{\fincv{\toPPCX[V][M]{u}}}
    = \fincv{\toPPCX[V][M]{\sigma}[\theta]}$, hence it is decided.
    Then, $\toPPCX[V][M]{t} \reduce[\rPPCX]
    \fdecv{\matchapply{\match[\fsize{\theta}]{\toPPCX[V][\listconc{\theta}{M}]{p}}{\fincv{\toPPCX[V][M]{u}}}}{\toPPCX[\listconc{\theta}{V}][M]{s}}}
    =
    \fdecv{\subsapply{\fincv{\toPPCX[V][M]{\sigma}[\theta]}}{\toPPCX[\listconc{\theta}{V}][M]{s}}}
    =
    \fdecv{\subsapply{\subs{\termidxv{1}[j]}{\fincv{\toPPCX[V][M]{\subsapply{\sigma}{\termvar{x_j}}}}}_{x_j \in \dom{\sigma}}}{\toPPCX[\listconc{\theta}{V}][M]{s}}}$.
    We conclude by Lem.~\ref{l:translation:to-ppcx-subs} with $i = 1$, since
    $\fdecv{\subsapply{\subs{\termidxv{1}[j]}{\fincv{\toPPCX[V][M]{\subsapply{\sigma}{\termvar{x_j}}}}}_{x_j \in \dom{\sigma}}}{\toPPCX[\listconc{\theta}{V}][M]{s}}}
    = \toPPCX[V][M]{\subsapply{\sigma}{s}} = \toPPCX[V][M]{t'}$.
    
    \item $\match[\theta]{p}{u} = \matchfail$. Then, $t' =
    \matchapply{\matchfail}{s} =
    \termabs[\lista{x}]{\termmatch{x}}{\termvar{x}}$. By
    Lem.~\ref{l:translation:to-ppcx-match}
    (\ref{l:translation:to-ppcx-match:fail}),
    $\toPPCX[V][M]{\match[\theta]{p}{u}} =
    \match[\fsize{\theta}]{\toPPCX[V][\listconc{\theta}{M}]{p}}{\toPPCX[V][M]{u}}
    = \matchfail$. Moreover, by Lem.~\ref{l:bisimulation:fincv-lifting},
    $\match[\fsize{\theta}]{\toPPCX[V][\listconc{\theta}{M}]{p}}{\fincv{\toPPCX[V][M]{u}}}
    = \matchfail$ too. Thus, $\toPPCX[V][M]{t} \reduce[\rPPCX]
    \matchapply{\matchfail}{\toPPCX[V][M]{s}} =
    \termabs[1]{\termidxm{1}[1]}{\termidxv{1}[1]} = \toPPCX[V][M]{t'}$. Hence,
    we conclude.
  \end{enumerate}
  
  \item $\ctxt{C} = \termapp{\ctxt{C'}}{u'}$. Then, $t = \termapp{r}{u'}$ and 
  $t' = \termapp{r'}{u'}$ with $r \reduce[\rPPC] r'$. By \ih $\toPPCX[V][M]{r}
  \reduce[\rPPCX] \toPPCX[V][M]{r'}$. Finally we conclude by
  Def~\ref{d:translation:to-ppcx}, since $\toPPCX[V][M]{t} =
  \termapp{\toPPCX[V][M]{r}}{\toPPCX[V][M]{u'}} \reduce[\rPPCX]
  \termapp{\toPPCX[V][M]{r'}}{\toPPCX[V][M]{u'}} = \toPPCX[V][M]{t'}$.
  
  \item $\ctxt{C} = \termapp{s'}{\ctxt{C'}}$. Then, $t = \termapp{s'}{r}$ and 
  $t' = \termapp{s'}{r'}$ with $r \reduce[\rPPC] r'$. By \ih $\toPPCX[V][M]{r}
  \reduce[\rPPCX] \toPPCX[V][M]{r'}$. Finally we conclude by
  Def~\ref{d:translation:to-ppcx}, since $\toPPCX[V][M]{t} =
  \termapp{\toPPCX[V][M]{s'}}{\toPPCX[V][M]{r}} \reduce[\rPPCX]
  \termapp{\toPPCX[V][M]{s'}}{\toPPCX[V][M]{r'}} = \toPPCX[V][M]{t'}$.
  
  \item $\ctxt{C} = \termabs[\theta']{\ctxt{C'}}{s'}$. Then, $t =
  \termabs[\theta']{q}{s'}$ and $t' = \termabs[\theta']{q'}{s'}$ with $q
  \reduce[\rPPC] q'$. By \ih $\toPPCX[V][\listconc{\theta'}{M}]{q}
  \reduce[\rPPCX] \toPPCX[V][\listconc{\theta'}{M}]{q'}$. Finally we conclude
  by Def~\ref{d:translation:to-ppcx}, since $\toPPCX[V][M]{t} =
  \termabs[\fsize{\theta'}]{\toPPCX[V][\listconc{\theta'}{M}]{q}}{\toPPCX[\listconc{\theta'}{V}][M]{s'}}
  \reduce[\rPPCX]
  \termabs[\fsize{\theta'}]{\toPPCX[V][\listconc{\theta'}{M}]{q'}}{\toPPCX[\listconc{\theta'}{V}][M]{s'}}
  = \toPPCX[V][M]{t'}$.
  
  \item $\ctxt{C} = \termabs[\theta']{p'}{\ctxt{C'}}$. Then, $t =
  \termabs[\theta']{p'}{r}$ and $t' = \termabs[\theta']{p'}{r'}$ with $r
  \reduce[\rPPC] r'$. By \ih $\toPPCX[\listconc{\theta'}{V}][M]{r}
  \reduce[\rPPCX] \toPPCX[\listconc{\theta'}{V}][M]{r'}$. Finally we conclude
  by Def~\ref{d:translation:to-ppcx}, since $\toPPCX[V][M]{t} =
  \termabs[\fsize{\theta'}]{\toPPCX[V][\listconc{\theta'}{M}]{p'}}{\toPPCX[\listconc{\theta'}{V}][M]{r}}
  \reduce[\rPPCX]
  \termabs[\fsize{\theta'}]{\toPPCX[V][\listconc{\theta'}{M}]{p'}}{\toPPCX[\listconc{\theta'}{V}][M]{r'}}
  = \toPPCX[V][M]{t'}$.
\end{itemize}
\end{proof}

%%% Local Variables: 
%%% mode: latex
%%% TeX-master: "main"
%%% End: 

\begin{theorem}
Let $t \in \Term{\calcPPC}$. If $t \reduce[\rPPC] t'$, then $\toPPCX{t}
\reduce[\rPPCX] \toPPCX{t'}$.
\label{t:bisimulation:ppc-ppcx}
\end{theorem}

\begin{proof}
The property is an immediate consequence of Lem.~\ref{l:bisimulation:ppc-ppcx},
taking $V = M = X$ as given in Def.~\ref{d:translation:to-ppcx}.
\end{proof}

Regarding the converse simulation, \ie $\calcPPCX$ into $\calcPPC$, we resort
here to the fact that the embedding $\toPPC{\_}$ also preserves the matching
operation (\cf Lem.~\ref{l:translation:to-ppc-match}). Then, every redex in
$\calcPPCX$ is translated into a redex in $\calcPPC$ as well.

\begin{lemma}
Let $t \in \Term{\calcPPCX}$. If $t \reduce[\rPPCX] t'$, then $\toPPC[V][M]{t}
\reduce[\rPPC] \toPPC[V][M]{t'}$.
\label{l:bisimulation:ppcx-ppc}
\end{lemma}

\begin{proof}
By definition $t \reduce[\rPPCX] t'$ implies $t =
\ctxtapply{\ctxt{C}}{\termapp{(\termabs[n]{p}{s})}{u}}$ and $t' =
\ctxtapply{\ctxt{C}}{\fdecv{\matchapply{\match[n]{p}{\fincv{u}}}{s}}}$ with
$\match[n]{p}{\fincv{u}}$ decided. We proceed by induction on $\ctxt{C}$.
\begin{itemize}
  \item $\ctxt{C} = \Box$. Then, $t = \termapp{(\termabs[n]{p}{s})}{u}$
  and $t' = \fdecv{\matchapply{\match[n]{p}{\fincv{u}}}{s}}$ with
  $\match[n]{p}{\fincv{u}}$ a decided match. Moreover, $\toPPC[V][M]{t} =
  \termapp{(\termabs[\theta]{\toPPC[V][\listconc{\theta}{M}]{p}}{\toPPC[\listconc{\theta}{V}][M]{s}})}{\toPPC[V][M]{u}}$
  with $\theta$ a list of $n$ fresh symbols.  There are two possible cases:
  \begin{enumerate}
    \item $\match[n]{p}{\fincv{u}} = \sigma$. Then, $t' = \fdecv{\subsapply{\sigma}{s}}$
    and, by Lem.~\ref{l:bisimulation:fincv-lifting}, we have $\sigma =
    \fincv{\match[n]{p}{u}}$. Moreover, by Lem~\ref{l:appendix:finc}
    (\ref{l:appendix:finc:fdec-l}), $\fdecv{\sigma} =
    \fdecv{\fincv{\match[n]{p}{u}}} = \match[n]{p}{u}$. By
    Lem.~\ref{l:translation:to-ppc-match}
    (\ref{l:translation:to-ppc-match:subs}),
    $\match[\theta]{\toPPC[V][\listconc{\theta}{M}]{p}}{\toPPC[V][M]{u}} =
    \toPPC[V][M]{\match[n]{p}{u}}[\theta] =
    \toPPC[V][M]{\fdecv{\sigma}}[\theta] =
    \subs{\theta_j}{\toPPC[V][M]{\fdecv{\subsapply{\sigma}{\termidxv{1}[j]}}}}_{\termidx{1}[j] \in \dom{\sigma}}$.
    Then, by Lem.~\ref{l:translation:to-ppc-subs} with $i = 1$, we get the
    reduction $\toPPC[V][M]{t} \reduce[\rPPC]
    \matchapply{\match[\theta]{\toPPC[V][\listconc{\theta}{M}]{p}}{\toPPC[V][M]{u}}}{\toPPC[\listconc{\theta}{V}][M]{s}}
    =
    \subsapply{\subs{\theta_j}{\toPPC[V][M]{\fdecv{\subsapply{\sigma}{\termidxv{1}[j]}}}}_{\termidx{1}[j] \in \dom{\sigma}}}{\toPPC[\listconc{\theta}{V}][M]{s}}
    \eqalpha
    \toPPC[V][M]{\fdecv{\subsapply{\subs{\termidx{1}[j]}{\fincv{\fdecv{\subsapply{\sigma}{\termidxv{1}[j]}}}}_{\termidx{1}[j] \in \dom{\sigma}}}{s}}}$.
    Note that, by Lem~\ref{l:appendix:finc} (\ref{l:appendix:finc:fdec-r}), we
    have
    $\subs{\termidx{1}[j]}{\fincv{\fdecv{\subsapply{\sigma}{\termidxv{1}[j]}}}}_{\termidx{1}[j] \in \dom{\sigma}}
    = \sigma$. Thus, we conclude since $\toPPC[V][M]{t} \reduce[\rPPC]
    \toPPC[V][M]{\fdecv{\subsapply{\sigma}{s}}} = \toPPC[V][M]{t'}$.
    
    \item $\match[n]{p}{\fincv{u}} = \matchfail$. Then, $t' =
    \fdecv{\matchapply{\matchfail}{s}} =
    \termabs[1]{\termidxm{1}[1]}{\termidxv{1}[1]}$ since there are no free
    indices in the term. By Lem.~\ref{l:bisimulation:fincv-lifting},
    $\match[n]{p}{\fincv{u}} = \fincv{\match[n]{p}{u}} = \matchfail$. Moreover,
    by Lem.~\ref{l:translation:to-ppc-match}
    (\ref{l:translation:to-ppc-match:fail}),
    $\toPPC[V][M]{\match[n]{p}{u}}[\theta] =
    \match[\theta]{\toPPC[V][\listconc{\theta}{M}]{p}}{\toPPC[V][M]{u}}
    = \matchfail$. Thus, $\toPPC[V][M]{t} \reduce[\rPPC]
    \matchapply{\matchfail}{\toPPC[V][M]{s}} =
    \termabs[\lista{x}]{\termmatch{x}}{\termvar{x}} \eqalpha \toPPC[V][M]{t'}$.
    Hence, we conclude.
  \end{enumerate}
  
  \item $\ctxt{C} = \termapp{\ctxt{C'}}{u'}$. Then, $t = \termapp{r}{u'}$ and 
  $t' = \termapp{r'}{u'}$ with $r \reduce[\rPPC] r'$. By \ih $\toPPC[V][M]{r}
  \reduce[\rPPC] \toPPC[V][M]{r'}$. Finally we conclude by
  Def~\ref{d:translation:to-ppc}, since $\toPPC[V][M]{t} =
  \termapp{\toPPC[V][M]{r}}{\toPPC[V][M]{u'}} \reduce[\rPPC]
  \termapp{\toPPC[V][M]{r'}}{\toPPC[V][M]{u'}} = \toPPC[V][M]{t'}$.
  
  \item $\ctxt{C} = \termapp{s'}{\ctxt{C'}}$. Then, $t = \termapp{s'}{r}$ and 
  $t' = \termapp{s'}{r'}$ with $r \reduce[\rPPC] r'$. By \ih $\toPPC[V][M]{r}
  \reduce[\rPPC] \toPPC[V][M]{r'}$. Finally we conclude by
  Def~\ref{d:translation:to-ppc}, since $\toPPC[V][M]{t} =
  \termapp{\toPPC[V][M]{s'}}{\toPPC[V][M]{r}} \reduce[\rPPC]
  \termapp{\toPPC[V][M]{s'}}{\toPPC[V][M]{r'}} = \toPPC[V][M]{t'}$.
  
  \item $\ctxt{C} = \termabs[m]{\ctxt{C'}}{s'}$. Then, $t = \termabs[m]{q}{s'}$
  and $t' = \termabs[m]{q'}{s'}$ with $q \reduce[\rPPC] q'$. Let $\theta'$ be a
  list of $m$ fresh symbols. By \ih $\toPPC[V][\listconc{\theta'}{M}]{q}
  \reduce[\rPPC] \toPPC[V][\listconc{\theta'}{M}]{q'}$. Finally we conclude
  by Def~\ref{d:translation:to-ppc}, since $\toPPC[V][M]{t} \eqalpha
  \termabs[\theta']{\toPPC[V][\listconc{\theta'}{M}]{q}}{\toPPC[\listconc{\theta'}{V}][M]{s'}}
  \reduce[\rPPC]
  \termabs[\theta']{\toPPC[V][\listconc{\theta'}{M}]{q'}}{\toPPC[\listconc{\theta'}{V}][M]{s'}}
  \eqalpha \toPPC[V][M]{t'}$.
  
  \item $\ctxt{C} = \termabs[m]{p'}{\ctxt{C'}}$. Then, $t = \termabs[m]{p'}{r}$
  and $t' = \termabs[m]{p'}{r'}$ with $r \reduce[\rPPC] r'$. Let $\theta'$ be a
  list of $m$ fresh symbols. By \ih $\toPPC[\listconc{\theta'}{V}][M]{r}
  \reduce[\rPPC] \toPPC[\listconc{\theta'}{V}][M]{r'}$. Finally we conclude
  by Def~\ref{d:translation:to-ppc}, since $\toPPC[V][M]{t} \eqalpha
  \termabs[\theta']{\toPPC[V][\listconc{\theta'}{M}]{p'}}{\toPPC[\listconc{\theta'}{V}][M]{r}}
  \reduce[\rPPC]
  \termabs[\theta']{\toPPC[V][\listconc{\theta'}{M}]{p'}}{\toPPC[\listconc{\theta'}{V}][M]{r'}}
  \eqalpha \toPPC[V][M]{t'}$.
\end{itemize}
\end{proof}

%%% Local Variables: 
%%% mode: latex
%%% TeX-master: "main"
%%% End: 

\begin{theorem}
Let $t \in \Term{\calcPPCX}$. If $t \reduce[\rPPCX] t'$, then $\toPPC{t}
\reduce[\rPPC] \toPPC{t'}$.
\label{t:bisimulation:ppcx-ppc}
\end{theorem}

\begin{proof}
The property is an immediate consequence of Lem.~\ref{l:bisimulation:ppcx-ppc},
taking $V = M = X$ as given in Def.~\ref{d:translation:to-ppc}.
\end{proof}

As already commented, these previous results may be combined to obtain a
strong bisimulation between the two calculi. The invertibility result allows
to define a relation between terms in $\calcPPC$ and $\calcPPCX$. Given $t \in
\Term{\calcPPCX}$ and $s \in \Term{\calcPPC}$, let us write $t \eqtrans s$
whenever $\toPPC{t} \eqalpha s$ and, therefore, $\toPPCX{s} = t$ by
Thm.~\ref{t:translation:inverse}. Then, the strong bisimulation result states
that whenever $t \eqtrans s$ and $t \reduce[\rPPCX] t'$, there exists a term
$s'$ such that $t' \eqtrans s'$ and $s \reduce[\rPPC] s'$, and the other way
around. Graphically: \[
\begin{array}{c@{\qquad\text{and}\qquad}c}
\begin{tikzcd}[ampersand replacement=\&]
t \arrow{d}[anchor=north,left]{\rPPCX}
  \&[-25pt] \eqtrans
  \&[-25pt] s \arrow[densely dashed]{d}[anchor=north,left]{\rPPC} \\
t'
  \&[-25pt] \eqtrans
  \&[-25pt] s'
\end{tikzcd}
&
\begin{tikzcd}[ampersand replacement=\&]
t \arrow[densely dashed]{d}[anchor=north,left]{\rPPCX}
  \&[-25pt] \eqtrans
  \&[-25pt] s \arrow{d}[anchor=north,left]{\rPPC} \\
t'
  \&[-25pt] \eqtrans
  \&[-25pt] s'
\end{tikzcd}
\end{array} \]

\begin{theorem}[Strong bisimulation]
The relation $\eqtrans$ is a strong bisimulation with respect to
the reduction relations $\reduce[\rPPC]$ and $\reduce[\rPPCX]$ respectively.
\label{t:bisimulation}
\end{theorem}

\begin{proof}
The proof follows immediately from Thm.~\ref{t:bisimulation:ppc-ppcx}
and~\ref{t:bisimulation:ppcx-ppc} above, and the invertibility result
(Thm.~\ref{t:translation:inverse}) given in Sec.~\ref{s:translation}.
\end{proof}

The importance of this result resides in the fact that it guarantees that
$\calcPPC$ and $\calcPPCX$ have exactly the same operational semantics. An
immediate consequence of this is the confluence of $\reduce[\rPPCX]$.

\begin{theorem}[Confluence]
The reduction relation $\reduce[\rPPCX]$ is confluent (CR).
\label{t:bisimulation:cr}
\end{theorem}

\begin{proof}
The result follows directly from Thm.~\ref{t:preliminaries:ppc:cr} and the
strong bisimulation:
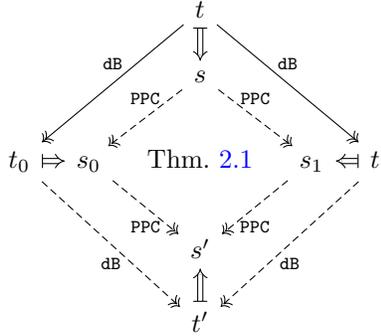
\begin{center}
\begin{tikzcd}[ampersand replacement=\&]
  \&[-15pt]  \&[-15pt] t \arrow[twoheadrightarrow]{ddll}[inner sep=5pt,above]{\rPPCX}
              \arrow[twoheadrightarrow]{ddrr}[inner sep=5pt,above]{\rPPCX}
              \arrow[Mapsto]{d}{}
\\[-5pt]
  \&[-15pt]  \&[-15pt] s \arrow[dashed,twoheadrightarrow]{dl}[inner sep=5pt,above]{\rPPC}
              \arrow[dashed,twoheadrightarrow]{dr}[inner sep=5pt,above]{\rPPC}
\\
t_0 \arrow[dashed,twoheadrightarrow]{ddrr}[inner sep=5pt,below]{\rPPCX}
    \arrow[Mapsto]{r}{}
  \&[-15pt] s_0 \arrow[dashed,twoheadrightarrow]{dr}[inner sep=5pt,below]{\rPPC}
  \&[-15pt] \text{Thm.~\ref{t:preliminaries:ppc:cr}}
  \&[-15pt] s_1 \arrow[dashed,twoheadrightarrow]{dl}[inner sep=5pt,below]{\rPPC}
  \&[-15pt] t_1 \arrow[dashed,twoheadrightarrow]{ddll}[inner sep=5pt,below]{\rPPCX}
                \arrow[Mapsto]{l}{}
\\
  \&[-15pt]  \&[-15pt] s'
\\[-5pt]
  \&[-15pt]  \&[-15pt] t' \arrow[Mapsto]{u}{}
\end{tikzcd}
\end{center}
\end{proof}

Another result of relevance to our line of research, is the existence of
normalising reduction strategies for $\calcPPC$, as it is shown
in~\cite{BonelliKLR17}. This result is particularly challenging since reduction
in $\calcPPC$ is shown to be \emph{non-sequential} due to the nature of its
matching operation. This implies that the notion of \emph{needed redexes} (a
key concept for defining normalising strategies) must be generalised to
\emph{necessary sets}~\cite{Raamsdonk97} of redexes. Moreover, the notion of
\emph{gripping}~\cite{Mellies96} is also captured by $\calcPPC$, representing a
further obstacle in the definition of such a normalising strategy. All in all,
in~\cite{BonelliKLR17} the authors introduce a reduction strategy $\S$ that is
shown to be normalising, overcoming all the aforementioned issues. Thanks to
the strong bisimulation result presented above, this strategy $\S$ can also be
guaranteed to normalise for $\calcPPCX$.

%%% Local Variables:
%%% mode: latex
%%% TeX-master: "main"
%%% End:

%%%%%%%%%%%%%%%%%%%%%%%%%%%%%%%%%%%%%%%%%%%%%%%%%%%%%%%%%%%%%%%%%%%%%%%%%%%%%%%
\section{Conclusion}
\label{s:conclusion}
%%%%%%%%%%%%%%%%%%%%%%%%%%%%%%%%%%%%%%%%%%%%%%%%%%%%%%%%%%%%%%%%%%%%%%%%%%%%%%%

In this paper we introduced a novel presentation of the \emph{Pure Pattern
Calculus} ($\calcPPC$)~\cite{JayK09} in de Bruijn's style. This required
extending de Bruijn ideas for a setting where each binder may capture more than
one variable at once. To this purpose we defined \emph{bidimensional indices}
of the form $\termidx{i}[j]$ where $\termidx{i}$ is dubbed the \emph{primary
index} and $\termidx{j}$ the \emph{secondary index}, so that the primary index
determines the binding abstraction and the secondary index identifies the
variable among those (possibly many) bound ones. Moreover, given the nature of
$\calcPPC$ semantics, our extension actually deals with two kinds of
bidimensional indices, namely \emph{variable indices} and \emph{matchable
indices}. This newly introduced calculus is simply called \emph{Pure Pattern
Calculus with de Bruijn indices} ($\calcPPCX$).

Our main result consists of showing that the relation between $\calcPPC$ and
$\calcPPCX$ is a strong bisimulation with respect to their respective redution
relations, \ie they have exactly the same operational semantics. For that
reason, proper translations between the two calculi were defined, $\toPPCX{\_}
: \Term{\calcPPC} \to \Term{\calcPPCX}$ and $\toPPC{\_} : \Term{\calcPPCX} \to
\Term{\calcPPC}$, in such a way that $\toPPCX{\_}$ is the inverse of
$\toPPC{\_}$ and vice-versa (modulo $\alpha$-conversion). Most notably, these
embeddings are shown to preserve the matching operation of their respective
domain calculus.

The strong bisimulation result allows to port into $\calcPPCX$ many already
known results for $\calcPPC$. Of particular interest for our line of research
are the confluence and the existence of normalising reduction strategies for
$\calcPPC$~\cite{BonelliKLR17}, a rather complex result that requires dealing
with notions of \emph{gripping}~\cite{Mellies96} and \emph{necessary
sets}~\cite{Raamsdonk97} of redexes. The result introduced on this paper may
allow for a direct implementation of such strategies without the inconveniences
of working modulo $\alpha$-conversion.

As commented before, the ultimate goal of our research is the implementation of
a prototype for a typed functional programming language capturing \emph{path
polymorphism}. This development is based on the \emph{Calculus of Applicative
Patterns} ($\calcCAP$)~\cite{Ayala-RinconBEV19}, for which a static type
system has already been introduced, guaranteeing well-behaved operational
semantics, together with its corresponding efficient type-checking algorithm.
$\calcCAP$ is essentially the static fragment of $\calcPPC$, where the
abstraction is generalised into an alternative (\ie abstracting multiple
branches at once). These two calculi are shown to be equivalent. Thus, future
lines of work following the results presented in this paper involve porting
$\calcCAP$ to the bidimensional indices setting and formalising the ideas
of~\cite{BonelliKLR17} into such framework. This will lead to a first
functional version of the sought-after prototype.

%%% Local Variables:
%%% mode: latex
%%% TeX-master: "main"
%%% End:

\bibliographystyle{alpha}
\bibliography{biblio}

\end{document}